\documentclass[11pt]{article}
\usepackage[text={17cm,24cm},centering]{geometry}
\usepackage{graphicx}
\usepackage{amsmath,amssymb,amsthm,upgreek}
\usepackage{MnSymbol,dsfont,mathbbol}
\usepackage[square,numbers]{natbib}
\usepackage{authblk}
\usepackage[format=hang,indention=-3em,textfont={small,it},labelfont={small,bf},width=0.92\textwidth]{caption}
\usepackage{enumitem}
\usepackage[notref,final]{showkeys}  % do not show labels

\usepackage{graphicx}
\graphicspath{{./figures/}}

\usepackage{siunitx}
\newtheorem{prop}{Proposition}
\newtheorem{remark}{Remark}
\newtheorem{lemma}{Lemma}
\newtheorem{theorem}{Theorem}

\usepackage{amssymb,MnSymbol,amsmath}
\usepackage{accents}

\newcommand{\be}{\boldsymbol{e}}
\newcommand{\bk}{\boldsymbol{k}}
\newcommand{\bu}{\boldsymbol{u}}
\newcommand{\br}{\boldsymbol{r}}
\newcommand{\bs}{\boldsymbol{s}}
\newcommand{\bt}{\boldsymbol{t}}
\newcommand{\bv}{\boldsymbol{v}}
\newcommand{\bi}{\boldsymbol{i}}
\newcommand{\bsig}{\boldsymbol{\sigma}}
\newcommand{\bI}{\boldsymbol{I}}
\newcommand{\bmI}{\boldsymbol{\mathcal{I}}}
\newcommand{\bE}{\boldsymbol{E}}
\newcommand{\bH}{\boldsymbol{H}}
\newcommand{\bzero}{\boldsymbol{0}}
\newcommand{\bLam}{\boldsymbol{\Lambda}}

\newcommand{\langx}{\mathbb{\Langle}}
\newcommand{\rangx}{\mathbb{\Rangle}}
\newcommand{\shh}{\!=\!}

\newcommand{\bT}{\boldsymbol{T}}
\newcommand{\bR}{\boldsymbol{R}}
\newcommand{\bmT}{\boldsymbol{\mathcal{T}}}
\newcommand{\bmR}{\boldsymbol{\mathcal{R}}}
\newcommand{\bmII}{\boldsymbol{\mathcal{I}}}

\newcommand{\bxi}{\boldsymbol{\xi}}

\newcommand{\bzeta}{\boldsymbol{\zeta}}
\newcommand{\bbeta}{\boldsymbol{\beta}}

\newcommand{\leftw}[1]{\overleftharpoon{#1}}
\newcommand{\rightw}[1]{\overrightharpoon{#1}}

\newcommand{\tran}{^{\text{\tiny T}}}

\newcommand{\sip}{\!\cdot\!}
\newcommand{\hh}{\hspace*{0.3pt}}
\newcommand{\nn}{\hspace*{-0.3pt}}
\newcommand{\es}{\hspace*{0.7pt}}
\newcommand{\nes}{\hspace*{-0.7pt}}
\newcommand{\dY}{\,\text{d}Y_{\nes\bxi}}

\newcommand{\Ereg}{\boldsymbol{\mathcal{E}}}

\newcommand{\jsup}[1]{^{\hh\text{\tiny [#1]}}}

\begin{document}
%----------------------------------------------------------------------------------------------------------------------%

\title{Plug-and-play analytical paradigm for the scattering of plane waves by ``layer-cake'' periodic systems}

\author[1]{Prasanna Salasiya}
\author[2]{Shixu Meng}
\author[1]{Bojan B. Guzina\thanks{Corresponding author (guzin001@umn.edu)}} % \ead{guzin001@umn.edu}
\affil[1]{\small{Civil, Environmental, \& Geo- Engineering, University of Minnesota Twin Cities, Minneapolis, MN, U.S.}}
\affil[2]{{\small Department of Mathematics, Virginia Tech, Blacksburg, VA 24061, U.S.}}

\date{}

\maketitle

\begin{abstract}
\noindent We investigate the scattering of scalar plane waves in two dimensions by a heterogeneous layer that is periodic in the direction parallel to its boundary. On describing the layer as a union of periodic laminae, we develop a solution of the scattering problem by merging the concept of propagator matrices and that of Bloch eigenstates featured by the unit cell of each lamina. The featured Bloch eigenstates are obtained by solving the quadratic eigenvalue problem (QEVP) that seeks a complex-valued wavenumber normal to the layer boundary given (i) the excitation frequency, and (ii) real-valued wavenumber parallel to the boundary -- that is preserved throughout the system. Spectral analysis of the QEVP reveals sufficient conditions for discreteness of the eigenvalue spectrum and the fact that all eigenvalues come in complex-conjugate pairs. By deploying the factorization afforded by the propagator matrix approach, we demonstrate that the contribution of individual eigenvalues (and so eigenmodes) to the solution diminishes exponentially with absolute value of their imaginary part, which then forms a rational basis for truncation of the factorized Bloch-wave solution. The proposed methodology caters for the optimal design of rainbow traps, energy harvesters, and metasurfaces, whose potency to manipulate waves is decided not only by the individual dispersion characteristics of the component laminae, but also by ordering and generally fitting of the latter into a composite layer. Using the factorized Bloch approach, evaluation of trial configurations -- as generated by the permutation and window translation/stretching of the component laminae -- can be accelerated by decades.
\end{abstract}

\maketitle 

%----------------------------------------------------------------------------------------------------------------------%
\section{Introduction}
%----------------------------------------------------------------------------------------------------------------------%

Since the pioneering work by Veselago~\citep{Vesel1964}, a mounting variety of metamaterial structures have been developed \citep[e.g.][]{Liu2011,Cumme2016,Kadi2019} to manipulate waves in the sub-wavelength regime, where the wavelength exceeds the characteristic ``microscopic'' lengthscale of an engineered composite. In most situations, metamaterials are designed -- for the purposes of e.g.  energy harvesting, vibration isolation, sensing, communication, and cloaking~\citep{jiao2023} -- through the prism of their spectral characteristics (specifically the frequency vs.~wavenumber map) which postulates an infinite periodic medium. Yet, their utility as wave manipulation tools is realized via \emph{bounded subsets} of such infinite constructs and their interaction with (each other and) the environment; keen examples include attenuation via rainbow traps \citep{Zhu2013,Tian2017}, redirection by metasurfaces \citep{Chen2016,Asso2018}, and near-lossless steering by topological wave protection~\citep{Zane2020}. As a result, the gap between system performance and spectral characteristics of its meta-components -- that is a key hurdle in the way of optimal design -- is normally bridged by either cyber of physical simulations. Such is the case for instance with rainbow trapping, where it was found~\citep{celli2019,chap2020} that the spatial permutation of units within a system may notably elevate the system performance. Nonetheless, the high cost of underpinning computer (or physical) simulations severely constricts the affordable design space and so limits the potency of metamaterial-based wave manipulation tools.

In this vein, our work focuses on the scattering by ``layer-cake'' periodic systems -- relevant to e.g. rainbow trapping and diffraction by metasurfaces -- and the way in which the Bloch wave analysis of the component periodic media enables \emph{modular}, plug-and-play evaluation of the system performance and so computationally-efficient optimal design. Among early treatments of the plane-wave reflection and transmission at interfaces between a homogeneous and periodic half-space via the Bloch wave paradigm are the works by Laude and coworkers~\citep{laude2009evanA,laude2009evanB}, which include consideration of corrugated interfaces~\citep{laude2011bloch}. A rigorous analysis of the radiation condition relevant to this class of problems can be found in~\citep{lama2018,fliss2016}. More recently, the featured reflection-transmission problem was revisited in a systematic manner \citep[e.g.][]{Kulpe2014,Kulpe2015,Sriva2020} that features Fourier series expansion of the interfacial conditions, applied to suitably factored (incident, reflected, and transmitted) wavefields whose traces on the contact surface follow periodicity of the heterogeneous half-space. Extensions of the joined-half-spaces problem to scattering by a strip of a periodic medium (see Fig.~\ref{strips}(b)) can be found in \citep{kulpe2014det,Kulpe2015}. Rainbow trapping by a finite sequence of \emph{graded} (C-shaped) resonator arrays was recently investigated in~\citep{Benn2019}; we shall revisit this study shortly. 

To better understand the mechanics of wave control by -- and help the optimal design of -- ``flat'' periodic meta-structures, we focus on the scattering by a periodic layer designed as a one-dimensional sequence of \emph{bonded strips} of dissimilar 2D periodic media, see Fig.~\ref{strips}(c). In this setting, we develop a solution of the scattering problem by deploying the concept of propagator matrices where the Bloch eigenstates of each member periodic medium serve as the carriers of phase and energy through the respective strip. Here, the Bloch eigenstates solve the quadratic eigenvalue problem (QEVP) that seeks a complex-valued wavenumber~$k_1$ normal to the strip boundary given (i) the excitation frequency~$\omega$, and (ii) real-valued wavenumber~$k_2$ parallel to the boundary -- that must be preserved throughout the system~\citep{lamacz2018outgoing}. Despite its central role to the problem, however, little is known \citep{Kulpe2014,Sriva2020} about the spectral characteristics of the featured QEVP. Given the newly obtained result that the eigenvalues $k_1(\omega,k_2)$ come in complex-conjugate pairs which allows for both exponential growth and decay of the solution in the 1-direction, we express the wave motion in the strip in terms of the Bloch waves that \emph{propagate} (for~$k_1$ real) or \emph{decay} (for~$k_1$ complex) either ``to the left'' or ``to the right''; a feature that caters for the development of a stable factorized Bloch scheme featuring well-conditioned transfer matrices. Thanks to the factorization afforded by such propagator approach, we demonstrate explicitly that the contribution of individual eigenvalues (and respective eigenmodes) of individual strips to the global solution diminishes exponentially with the absolute value of the imaginary part of~$k_1$, which then forms a rational basis for truncation of the factorized Bloch-wave solution. In this way, evaluation of trial layer-cake configurations targeting given 2D functionality (e.g. wave attenuation or redirection)  -- as generated by the permutation and window translation/stretching of the component strips -- can be accelerated by decades for it entails only low-dimensional (propagator) matrix manipulations. 

For completeness, we note that the method of transfer matrices in the context of rainbow trapping was recently deployed in~\citep{Benn2019}. There, the wavefield in each strip (an infinite stack of C-shaped diffraction gratings) was resolved via plane-wave expansion that lends itself to the solution of a linear eigenvalue problem seeking $\omega(k_1,k_2)$ for~$k_{1/2}$ real. In such framework, evanescent Bloch waves signifying the boundary layers due to strip interfaces are implicitly accounted for by way of transfer matrices featuring complex eigenvalues. By contrast, the present study relies on the Bloch-wave expansion that, via the linchpin QEVP, exposes both propagating and evanescent Bloch waves from the onset. By providing such natural basis for the expansion of wavefields in layer-cake periodic systems, the proposed solution can be viewed as a congenital model order reduction~\citep{Schi2008} that features fast convergence and modular evaluation structure catering for the optimal design of this class of wave manipulation devices. 

%----------------------------------------------------------------------------------------------------------------------%
\section{Wave motion in an infinite periodic medium}\label{trans}
%----------------------------------------------------------------------------------------------------------------------%

With reference to a Cartesian coordinate system endowed with an orthonormal vector basis~$\be_j$ ($j\shh \overline{1,2}$), consider a two-dimensional scalar wave equation 
\begin{equation}\label{2DSW}
\nabla \sip (G(\bxi) \nabla u) + \omega^2\rho(\bxi)u \:=\: 0, \qquad \bxi \in \mathbb{R}^2
\end{equation}
at frequency~$\omega$, where $G$ and $\rho$ are $Y$-periodic; 
\[
Y = \{\bxi: 0 < \nes\xi_1\nes < \ell, \: 0 < \nes\xi_2\nes < d \}, \qquad \xi_j = \bxi\sip\be_j 
\]
is the unit cell of periodicity (see Fig.\ref{strips}(a)), and the wavefield $u(\bxi)$ carries implicit time dependence~$e^{-i\omega t}$. Hereon, we assume that $G$ and~$\rho$ are real-valued $L^\infty(Y)$ functions bounded from below away from zero in that $G_{\inf} \leqslant G \leqslant G_{\sup}$ and $\rho_{\inf} \leqslant \rho \leqslant \rho_{\sup}$ for some positive constants $G_{\inf}, G_{\sup}, \rho_{\inf}$ and~$\rho_{\sup}$. In the case of horizontally-polarized (SH) elastic wave motion -- that is assumed here for the sake of discussion, $G,\rho$ and $u$ signify respectively the shear modulus, mass density, and transverse displacement. When interpreted in the context of acoustics, on the other hand $G,\rho$ and $u$ stand respectively for the reciprocal mass density, reciprocal bulk modulus, and pressure.  Recalling the Floquet-Bloch theorem \citep{Floq1883,Bloch1929}, we write
\begin{equation} \label{FBT}
u(\bxi) \:=\: \phi(\bxi) e^{i \bk \cdot \bxi}, \qquad \phi:\: \text{$Y$-periodic}
\end{equation}
where $\bk\in\mathbb{C}^2$ is the germane wave vector and $\phi$ depends implicitly on~$\omega$ and~$\bk$. On substituting~\eqref{FBT} in~\eqref{2DSW}, we obtain 
\begin{align} \label{2DFB}
\begin{aligned}
\nabla_{\!\bk}  \sip (G(\bxi) &\nabla_{\!\bk} \phi) + \omega^2 \rho(\bxi) \phi = 0, \quad \bxi \in Y \\*[1mm]
\phi|_{\xi_j=0} = \phi|_{\xi_j=\ell_j}, &\quad
\boldsymbol{n} \sip G\nabla_{\!\bk}\phi|_{\xi_j=0} = -\boldsymbol{n} \sip G\nabla_{\!\bk}\phi|_{\xi_j=\ell_j}, \quad j\shh \overline{1,2}
 \end{aligned}
\end{align}
where $\nabla_{\!\bk}  = \nabla \!+ {i} \bk$; $\,\hh \ell_1\!=\ell$, $\,\hh \ell_2\!=d$,  and $\boldsymbol{n}$ is the unit outward normal on~$\partial{Y}$. 

\begin{figure}
\centering{\includegraphics [width=1.0\textwidth ]{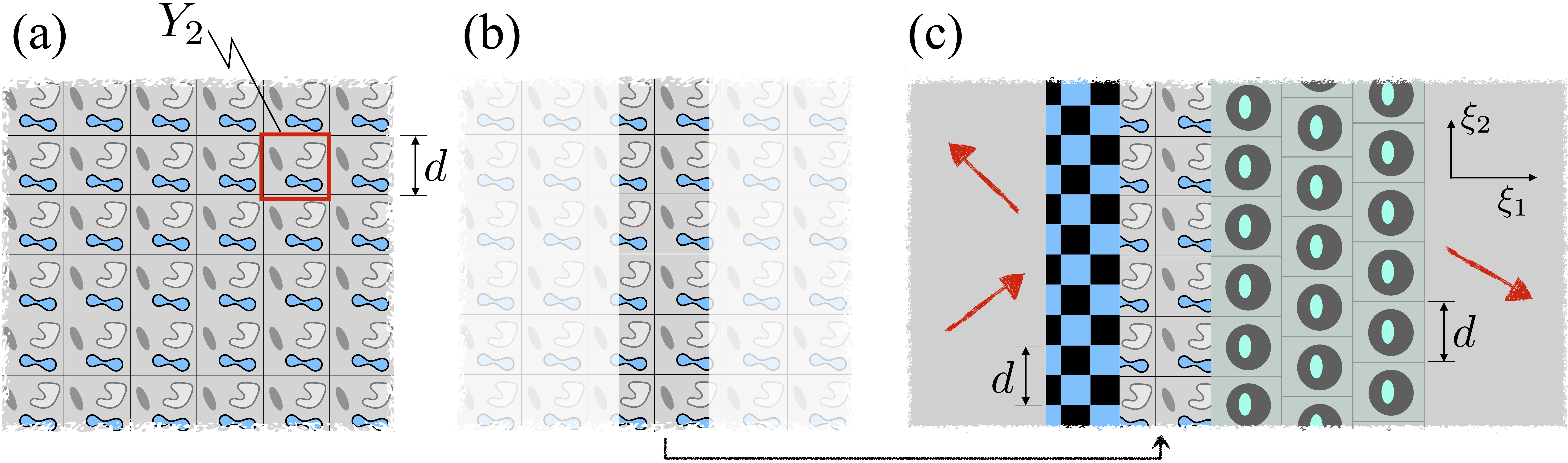}} 
\caption{(a) 2D periodic medium $\mathcal{M}_2$ endowed with unit cell~$Y_2$; (b) strip of $\mathcal{M}_2$, and (c) periodic layer~$\Omega$ comprising five bonded strips cut out of periodic media $\mathcal{M}_q$ ($q=\overline{1,3}$) that are each supported on an orthogonal lattice and share the lengthscale of periodicity, $d$, in the $\xi_2$-direction. The last three laminae of~$\Omega$ are constructed by the $\xi_2$-translation of a strip of~$\mathcal{M}_3$.}\label{strips}
\end{figure}
%----------------------------------------------------------------------------------------------------------------------%
\subsection{Quadratic eigenvalue problem (QEVP)} \label{SecQEVP}
%----------------------------------------------------------------------------------------------------------------------%

Let~$\bk=(k_1\!:=\kappa,k_2)$, and consider the Hilbert spaces
\begin{eqnarray*}
L_p^2(Y)  &\!\shh \!&  \{g\in L^2(Y):\, (g, g) < \infty, \,\,g|_{x_j=0} = g|_{x_j=\ell_j}\}, \\
H^1_{p} (Y) &\! \shh \!& \{g\in L^2(Y):\, \nabla g \in (L^2(Y))^2, \,\,g|_{x_j=0} = g|_{x_j=\ell_j}\}
\end{eqnarray*}
where~$(g,h)=\int_{Y}g(\bxi)\es\overline{h}(\bxi)\dY$ and $\overline{h}$ is the complex conjugate of~$h$. 

To facilitate the ensuing developments, it is instructive to consider a quadratic eigenvalue problem~\citep{laude2009evanA,Sriva2020} that consists in seeking the eigenvalues~$\kappa_n(\omega, k_2)\in\mathbb{C}$ and eigenfunctions~$\phi_n(\bxi; \omega, k_2)\in H_p^1(Y)$ solving~\eqref{2DFB} for given $\omega\in\mathbb{R}^+$ and~$k_2\!\in\mathbb{R}$. Writing the featured quadratic  eigenvalue problem (QEVP) in weak form, we obtain 
\begin{multline} \label{QEP}
\kappa_n^2\!\! \int_Y G \phi_n\es \overline{\psi} \,\text{d}Y_{\bxi} \,-\, 
\kappa_n \left\{ i\!\! \int_Y G \left( \frac{\partial \phi_n}{\partial \xi_1} \overline{\psi} - \phi_n \frac{\partial \bar{\psi}}{\partial \xi_1} \right)\dY \right\} \,+\, 
\left\{ k_2^2 \!\int_Y G \phi_n\es\overline{\psi}\dY \right. - \\ \left. k_2 \es i \!\!\int_Y G \left( \frac{\partial \phi_n}{\partial \xi_2} \overline{\psi} - \phi_n \frac{\partial \bar{\psi}}{\partial \xi_2} \right)\dY \,-\, \omega^2\!\! \int_Y \rho \phi_n\es \overline{\psi} \dY + \!\int_Y G \nabla \phi_n \sip \nabla \overline{\psi} \dY \right\}  \:=\: 0, ~~~ \forall\psi(\bxi) \in H_p^1(Y).
\end{multline}
Depending on the nature of $\kappa_n$, field $u=\phi_n(\bxi)e^{i (\kappa_n,k_2) \cdot \bxi}$ signifies either a bulk (i.e.~Bloch) wave propagating in the $\bk$-direction (when $\kappa_n\!\in\mathbb{R}$), an evanescent wave propagating in the~$\xi_1$-direction (when $\kappa_n\!\in i\mathbb{R}$), or a ``leaky'' evanescent wave propagating in the~$\xi_1$-direction (when $\kappa_n\!\in\mathbb{C}$).

%----------------------------------------------------------------------------------------------------------------------%
\subsubsection{Spectral properties of the QEVP}
%----------------------------------------------------------------------------------------------------------------------%

Letting $\langle \cdot,\cdot \rangle$ denote the inner product affiliated with $H^1_{p} (Y)$, we introduce operators
$\mathcal{A}:H^1_{p} (Y)  \mapsto H^1_{p} (Y) $, $\, \mathcal{B}:H^1_{p}
(Y)  \mapsto H^1_{p} (Y) $ and $\, \mathcal{C}:H^1_{p} (Y)  \mapsto
H^1_{p} (Y)$ respectively by  
\begin{eqnarray}
\langle \mathcal{A} \phi, \psi \rangle &:=&  k_2^2 \!\int_Y G \phi\es\overline{\psi}\dY   - k_2 \es i \!\!\int_Y G \left( \frac{\partial \phi}{\partial \xi_2} \overline{\psi} - \phi \frac{\partial \bar{\psi}}{\partial \xi_2} \right)\dY \nonumber \\
&&\,-\, \omega^2\!\! \int_Y \rho \phi\es \overline{\psi} \dY + \!\int_Y G \nabla \phi \sip \nabla \overline{\psi} \dY, \quad \forall \phi,\psi \in H^1_{p} (Y),\\
\langle \mathcal{B} \phi, \psi \rangle &:=& -  i\!\! \int_Y G \left( \frac{\partial \phi}{\partial \xi_1} \overline{\psi} - \phi \frac{\partial \bar{\psi}}{\partial \xi_1} \right)\dY,\quad \forall \phi,\psi \in H^1_{p} (Y),\\
\langle \mathcal{C} \phi, \psi \rangle &:=& \int_Y G \phi\es \overline{\psi} \,\text{d}Y_{\bxi},\quad \forall \phi,\psi \in H^1_{p} (Y).
\end{eqnarray}
To aid the analysis, we also define an auxiliary map $\widetilde{\mathcal{A}}:H^1_{p} (Y)  \mapsto
H^1_{p} (Y) $ via 
\begin{multline}
\langle \widetilde{\mathcal{A}} \phi, \psi \rangle := \!\int_Y G \nabla \phi \sip \nabla \overline{\psi} \dY +  \left( |k_2|^2\frac{2G^2_{\sup}}{G_{\inf}}+1 \right) \!\! \int_Y   \phi\es \overline{\psi} \dY - k_2 \es i \!\!\int_Y G \left( \frac{\partial \phi}{\partial \xi_2} \overline{\psi} - \phi \frac{\partial \bar{\psi}}{\partial \xi_2} \right)\dY, \\ 
\quad \forall \phi,\psi \in H^1_{p} (Y).
\end{multline}
Using the above definitions, QEVP \eqref{QEP} can be rewritten as
\begin{equation} \label{QEP operator form}
 \mathcal{A} \phi_n \,+\, \kappa_n \hh \mathcal{B} \phi_n \,+\, \kappa_n^2 \hh \mathcal{C} \phi_n \,=\, 0.
\end{equation}
On introducing a quadratic operator pencil \citep{Markus1988} 
\begin{equation}\label{QEP operator pencil}
\mathcal{T}(\lambda) \,:=\, \mathcal{A}  + \lambda  \mathcal{B}   + \lambda^2  \mathcal{C},
\end{equation}
our aim to analyze the spectral properties of $\mathcal{T}(\lambda)$. To this end, we establish the following claims.

\begin{lemma}  \label{lemma1}
Operator $ \widetilde{\mathcal{A}}$  is coercive and satisfies the monotonicity relation
\begin{equation}
\widetilde{\mathcal{A}}_{\inf} \,\leqslant\, \widetilde{\mathcal{A}} \,\leqslant\,
\widetilde{\mathcal{A}}_{\sup}, 
\end{equation}
where
 $\widetilde{\mathcal{A}}_{\inf}:H^1_{p} (Y)  \mapsto H^1_{p} (Y) $ is given by
\begin{eqnarray} \label{A_inf variational}
\langle \widetilde{\mathcal{A}}_{\inf} \phi, \psi \rangle := \!\int_Y \frac{G_{\inf}}{2} \nabla \phi \sip \nabla \overline{\psi} \dY +   \!\! \int_Y   \phi\es \overline{\psi} \dY, \quad \forall \phi,\psi \in H^1_{p} (Y)
\end{eqnarray}
and $\, \widetilde{\mathcal{A}}_{\sup}:H^1_{p} (Y)  \mapsto H^1_{p} (Y) $ satisfies
\begin{eqnarray}
\langle \widetilde{\mathcal{A}}_{\sup} \phi, \psi \rangle := \!\int_Y \big(G_{\sup}+\frac{G_{\inf}}{2} \big) \nabla \phi \sip \nabla \overline{\psi} \dY +  \left( 2|k_2|^2\frac{2G^2_{\sup}}{G_{\inf}}+1 \right) \!\! \int_Y   \phi\es \overline{\psi} \dY, ~~ \forall \phi,\psi \in H^1_{p} (Y).
\end{eqnarray}
\end{lemma}
\begin{proof}
If $k_2=0$, the monotonicity relation can be directly established. Hereon we prove the lemma assuming $k_2\not=0$. We first recall that $G_{\inf} \leqslant G \leqslant G_{\sup}$ for some positive $G_{\inf}$ and $G_{\sup}$. By Young's inequality, for any $\phi \in H^1_p(Y)$ we accordingly have
\[
\left|\int_Y G \left( \frac{\partial \phi}{\partial \xi_2} \overline{\phi} - \phi\frac{\partial \bar{\phi}}{\partial \xi_2} \right)\dY\right| \:\leqslant\: 2 G_{\sup} \left(\frac{1}{2\tau}\|\nabla \phi\|^2_{L^2(Y)} + \frac{\tau}{2}\|\phi\|^2_{L^2(Y)} \right), \quad \tau = \frac{2|k_2|G_{\sup}}{ G_{\inf}},
\]
provided that $\tau\not=0$, i.e. $k_2\not=0$. This demonstrates that for any $\phi \in H^1_p(Y)$
\begin{eqnarray*}
\langle \widetilde{\mathcal{A}} \phi, \phi \rangle \!\!\!&\geqslant&\!\!\! G_{\inf}\|\nabla \phi\|^2_{L^2(Y)} + \left(|k_2|^2\frac{2G^2_{\sup}}{G_{\inf}}+1 \right)\|\phi\|^2_{L^2(Y)} -|k_2|2 G_{\sup} \left(\frac{1}{2\tau}\|\nabla \phi\|^2_{L^2(Y)} + \frac{\tau}{2}\|\phi\|^2_{L^2(Y)} \right) \\
\!\!\!\!&=&\!\!\!\! \frac{G_{\inf}}{2}\|\nabla \phi\|^2_{L^2(Y)} + \|\phi\|^2_{L^2(Y)} \:=\: \langle \widetilde{\mathcal{A}}_{\inf} \phi, \phi \rangle,
\end{eqnarray*}  
and 
\begin{eqnarray*}
\langle \widetilde{\mathcal{A}} \phi, \phi \rangle \!\!\!\!&\leqslant&\!\!\!\! G_{\sup}\|\nabla \phi\|^2_{L^2(Y)} + \left(|k_2|^2\frac{2G^2_{\sup}}{G_{\inf}}+1 \right)\|\phi\|^2_{L^2(Y)} +|k_2|2 G_{\sup} \left(\frac{1}{2\tau}\|\nabla \phi\|^2_{L^2(Y)} + \frac{\tau}{2}\|\phi\|^2_{L^2(Y)} \right) \\
\!\!\!\!\!&=&\!\!\!\!\! \big(G_{\sup}+\frac{G_{\inf}}{2} \big)\|\nabla \phi\|^2_{L^2(Y)} +\left(2|k_2|^2\frac{2G^2_{\sup}}{G_{\inf}}+1 \right) \|\phi\|^2_{L^2(Y)} \:=\: \langle \widetilde{\mathcal{A}}_{\sup} \phi, \phi \rangle.
\end{eqnarray*}  
\end{proof}

\begin{lemma} \label{lemma2}
Operators $\mathcal{A}:H^1_{p} (Y)  \mapsto H^1_{p} (Y) $, $\, \mathcal{B}:H^1_{p} (Y)  \mapsto H^1_{p} (Y) $ and $\,\mathcal{C}:H^1_{p} (Y)  \mapsto H^1_{p} (Y) $ are self-adjoint.
\end{lemma}
\begin{proof}
Since $G$ is real-valued, operator $\mathcal{C}$ is clearly self-adjoint. We next show that $\mathcal{B}$ carries the same property. Indeed, for any $\phi,\psi \in H^1_{p} (Y)$ we have
 \begin{eqnarray*}
\langle \mathcal{B}^* \phi, \psi \rangle =\langle  \phi, \mathcal{B} \psi \rangle = \overline{\langle  \mathcal{B} \psi ,\phi\rangle }= \overline{-  i\!\! \int_Y G \left( \frac{\partial \psi}{\partial \xi_1} \overline{\phi} - \psi \frac{\partial \bar{\phi}}{\partial \xi_1} \right)\dY} =  i\!\! \int_Y G \left(  \frac{\partial \bar{\psi}}{\partial \xi_1}  \phi  - \overline{\psi}\frac{\partial \phi}{\partial \xi_1}  \right)\dY=\langle \mathcal{B} \phi, \psi \rangle.
\end{eqnarray*}
The self-adjointness of $\mathcal{A}$ can be demonstrated in an analogous way.

\end{proof}

%{\color{blue}For sake of clarity and technical reasons, we assume hereon in this subsection that $G$ is at least $C^1$-continuous.}
\begin{lemma}  \label{lemma3}
Operator $\mathcal{B}:H^1_{p} (Y)  \mapsto H^1_{p} (Y) $ is compact, and operator $\mathcal{T}(\lambda)$ is Fredholm of index zero.
\end{lemma}
\begin{proof}
We note that $ \mathcal{B} =\mathcal{B}_1 + \mathcal{B}_2$, where
 \begin{equation}
\langle \mathcal{B}_1 \phi, \psi \rangle \:=\: - i\!\! \int_Y G \frac{\partial \phi}{\partial \xi_1} \overline{\psi}  \dY,\quad  \langle \mathcal{B}_2 \phi, \psi \rangle \:=\:    i\!\! \int_Y G   \phi \frac{\partial \bar{\psi}}{\partial \xi_1} \dY,\quad \forall \phi,\psi \in H^1_{p} (Y).
\end{equation}
Since $\|\mathcal{B}_2 \phi\|_{H^1_p(Y)} \leqslant\| \phi\|_{L^2(Y)}$, we find that  $\mathcal{B}_2$ is a compact operator due to the compact
embedding from $H^1_p(Y)$ into $L^2(Y)$. We next interpret $\int_Y G \frac{\partial \phi}{\partial \xi_1} \overline{\psi}  \dY$ as the duality pairing between $H^{-1}_p(Y)$ and $H^{1}_p(Y)$ with $L^2_p(Y)$ as the pivoting space, which yields $\|\mathcal{B}_1 \phi\|_{H^1_p(Y)} \leqslant \|G \frac{\partial\phi}{\partial \xi_1}\|_{H^{-1}_p(Y)}$. As a result, $\mathcal{B}_1$ is a compact operator due to compact embedding from $L^2(Y)$ into $H^{-1}_p(Y)$. It can be demonstrated in a similar way that operators $\mathcal{A}-\widetilde{\mathcal{A}}$ and $\mathcal{C}$ are also compact. Lastly, we note that $\mathcal{T}(\lambda)$ is the sum of an invertible operator and a compact operator, whereby operator $\mathcal{T}(\lambda)$ is Fredholm of index zero.
\end{proof}

\begin{prop} \label{prop1}
If $\mathcal{T}(\lambda) \phi=0\, $ for some non-zero $\phi \in H^1_p(Y)$, then there must exist a nontrivial $\psi \in H^1_p(Y)$ such that $\mathcal{T}(\overline{\lambda}) \psi=0$. Hence all eigenvalues of~\eqref{QEP} come in complex-conjugate pairs.
\end{prop}
\begin{proof}
Let $\mathcal{T}(\lambda) \phi=0$ for some nontrivial $\phi \in H^1_p(Y)$, and assume that $\mathcal{T}(\overline{\lambda}) \psi\not=0$ for all non-zero $\psi \in H^1_p(Y)$. Then, since $\mathcal{T}(\overline{\lambda})$ is a Fredholm operator of index zero by Lemma~\ref{lemma3}, $\mathcal{T}(\overline{\lambda})$ is injective and hence invertible. Since $\mathcal{A}$, $\mathcal{B}$ and $\mathcal{C}$ are self-adjoint by Lemma~\ref{lemma2}, \eqref{QEP operator pencil} yields $\mathcal{T}(\overline{\lambda}) = \big(\mathcal{T}({\lambda})\big)^*$
whereby $\big(\mathcal{T}({\lambda})\big)^*$ is also invertible. This is a contradiction since $H^1_p(Y)=\mbox{Range} \big(\mathcal{T}({\lambda})\big)^*= \big(\mbox{Null} \mathcal{T}({\lambda})\big)^\perp \not=H^1_p(Y)$. 
\end{proof}

The next two claims focus on the discreteness of the eigenspectrum of~\eqref{QEP}, that is an implicit premise of the subsequent analysis. 
\begin{theorem}
Assume that there exists $\tau \in \mathbb{C}$ such that $\mathcal{T}(\tau)$ is injective. Then the eigenvalues $\lambda$ affiliated with the quadratic pencil $\mathcal{T}(\lambda)$  form a discrete set, with infinity being the only possible accumulation point.
\end{theorem}
\begin{proof}
We first note that $\mathcal{T}(\lambda) = \widetilde{\mathcal{A}} +  \left(\mathcal{A}-\widetilde{\mathcal{A}} \right) + \lambda  \mathcal{B}   + \lambda^2  \mathcal{C}$, where $\widetilde{\mathcal{A}}$ is invertible while  $\mathcal{A}-\widetilde{\mathcal{A}}$, $\mathcal{B}$ and $\mathcal{C}$ are compact. 
Since by premise there exists $\tau \in \mathbb{C}$ such that $\mathcal{T}(\tau)$ is injective and  $\mathcal{T}(\lambda)$ is Fredholm operator of index zero, we conclude that all eigenvalues $\lambda$  form a discrete set with infinity being the only possible accumulation point \cite[Theorem 12.9]{Markus1988}. 
\end{proof}

For a generic choice of~$G,\rho,\ell$ and~$\omega$, it is difficult (at this point) to establish the existence of~$\tau\!\in\mathbb{C}$ which ensures the injectivity of~$\mathcal{T}(\tau)$. However, there is a tractable class of ``low frequency'' configurations for which this condition is satisfied. 

\begin{lemma}
Let $\tau := \frac{\pi }{\ell}$. Then for all excitation frequencies bounded by $\omega < \tau \sqrt{G_{\inf}/\rho_{\sup}}$, we have 
\[
\Re\langle\mathcal{T}(\tau) \phi,\phi\rangle > 0 \quad \forall \phi \in H^1_{p} (Y),
\]
i.e. operator $\mathcal{T}(\tau)$ is injective.
\end{lemma}
\begin{proof}
For any $\phi \in H^1_{p} (Y)$, one can deploy the Fourier series expansion
$$
\phi(\xi_1,\xi_2) = \sum_{m=-\infty}^\infty \sum_{n=-\infty}^\infty a_{mn} e^{i \frac{2\pi m}{\ell} \xi_1 + i \frac{2\pi n}{d} \xi_2 }
$$
to obtain 
\begin{eqnarray*}
\!\int_Y G \left\| \left(\nabla + i k_2 \be_2 + i \tau \be_1 \right)\phi   \right\|^2 \dY \,\geqslant\, G_{\inf} \sum_{m=-\infty}^\infty \sum_{n=-\infty}^\infty |a_{mn}|^2   \left( \left(\frac{2\pi m}{\ell} + \tau\right)^2+  \left(\frac{2\pi m}{\ell} + k_2\right)^2 \right).
\end{eqnarray*}
Letting $\tau = \frac{\pi }{\ell}$, it follows that $|\frac{2\pi m}{\ell} + \tau|\geqslant \frac{\pi }{\ell}$ for all $m=0,\pm 1, \cdots$. From the above equation, we obtain 
\begin{eqnarray*}
\!\int_Y G \left\| \left(\nabla + i k_2 \be_2 + i \tau \be_1 \right)\phi   \right\|^2 \dY   \,\geqslant\, G_{\inf} \big(\frac{\pi }{\ell}\big)^2\sum_{m=-\infty}^\infty \sum_{n=-\infty}^\infty |a_{mn}|^2  = G_{\inf} \big(\frac{\pi }{\ell}\big)^2 \|\phi\|^2_{L^2_p(Y)}. 
\end{eqnarray*}  
On combining the last result with 
$$
\omega^2 \!\int_Y \rho |\phi|^2 \dY \,\leqslant\,  \omega^2 \rho_{\sup} \|\phi\|^2_{L^2_p(Y)}
$$
and letting $\tau = \tfrac{\pi }{\ell}$, we find that 
$$
\!\int_Y G \left\| \left(\nabla + i k_2 \be_2 + i \tau \be_1 \right)\phi   \right\|^2 \dY \,>\, \omega^2 \!\int_Y \rho |\phi|^2 \dY \quad\mbox{ if } \quad \omega < \tau \sqrt{G_{\inf}/\rho_{\sup}}.
$$
The proof is completed by observing that $\Re \langle \mathcal{T}(\tau) \phi,\phi \rangle = \!\int_Y G \left\| \left(\nabla + i k_2 \be_2 + i \tau \be_1 \right)\phi   \right\|^2 \dY -  \omega^2 \!\int_Y \rho |\phi|^2 \dY$,
which establishes the injectivity of $\mathcal{T}(\tau)$. 
\end{proof}

%----------------------------------------------------------------------------------------------------------------------%
\subsubsection{Essential strip of the $\kappa_n$-spectrum}\label{essential_strip}
%----------------------------------------------------------------------------------------------------------------------%

Let~$\kappa_n\!\in\mathbb{C}$ and $\phi_n\!\in H_p^1(Y)$ solve~\eqref{QEP} for given~$(\omega,k_2)$. Then the resulting Bloch wave~\eqref{FBT} can be written as 
\begin{equation} \label{com1}
\phi_n(\bxi) \, e^{i (\kappa_n,k_2) \cdot \bxi} \;=\; [\phi_n(\bxi)e^{-i 2\pi\xi_1/\ell}] \, e^{i (\kappa_n+2\pi/\ell,k_2) \cdot \bxi}. 
\end{equation}
Since the term inside the brackets is clearly $Y$-periodic, the right-hand side of~\eqref{com1} is a Bloch wave propagating at $k_1=\kappa_n\!+2\pi/\ell$. As a result, the solution pairs $(\kappa_n,\phi_n)$ and $(\kappa_n\!+2\pi/\ell,\phi_ne^{-i 2\pi\xi_1/\ell})$ are \emph{equivalent}. Indeed, for given~$(\omega,k_2)$, equations~\eqref{2DFB} remain invariant under the substitution
\begin{equation} \label{com2}
(\kappa_n,\phi_n) \;\mapsto\; (\kappa_n\!+2\pi/\ell,\phi_ne^{-i 2\pi\xi_1/\ell}).
\end{equation}
As a result, the eigenvalues of~\eqref{QEP} can be conveniently folded into the strip
 \begin{equation} \label{essenstrip}
\Re(\kappa_n) \in \Big[\!-\nes\frac{\pi}{\ell},\frac{\pi}{\ell}\Big),  \qquad \Im(\kappa_n) \in \mathbb{R}, 
\end{equation}
provided that the respective eigenfunctions are reckoned according to~\eqref{com2}.  

For clarity, we note that the $\kappa_n$-spectrum contains only a \emph{finite number of real eigenvalues}, $\kappa_n\!\nes \in\mathbb{R}$. To demonstrate this, we temporarily relabel the prescribed frequency~$\omega$ and wavenumber~$k_2$ in QEVP~\eqref{QEP} as $\omega'$ and $k_2'$, respectively. Then, with reference to the standard dispersion relationship $\omega_n(k_1,k_2)$ ($n\!\in\mathbb{Z}^+, \, k_1,k_2\!\in\mathbb{R}$) -- obtained by solving the linear eigenvalue problem due to~\eqref{2DFB}, $\kappa_n\!\in\mathbb{R}$ are given by the intersection of the dispersion surfaces $\omega_n(k_1,k_2)$ with the line specified by $\omega\shh \omega'$ and $k_2\shh k_2'$.

%----------------------------------------------------------------------------------------------------------------------%
\subsubsection{Truncated eigenspectrum} \label{trunceigsp}
%----------------------------------------------------------------------------------------------------------------------%

In the sequel, we assume discreteness of the set~$\{\kappa_n\}$ and recall Proposition~\ref{prop1} which demonstrates the symmetry of~$\{\kappa_n\}$ about the $\Re(\kappa_n)$-axis in the complex plane. In this vein, we shall conveniently consider a \emph{truncated spectrum}, $\{\kappa_n,\phi_n\}_{n=1}^{2N}$, of the featured QEVP such that:
\begin{itemize}

\item Subset $\{\kappa_n,\phi_n\}_{n=1}^{N}$ contains the eigenvalues with $N$~\emph{largest negative} imaginary parts, arranged in the order of ascending~$|\Im(\kappa_n)|$.  Concerning the eigenvalues $\kappa_n\!\in\mathbb{R}$, in this subset we include only those with $\kappa_n<0$, similarly arranged in the order of ascending~$|\kappa_n|$;
    
\item Subset $\{\kappa_n,\phi_n\}_{n=N+1}^{2N}$ contains the eigenvalues with $N$~\emph{smallest positive} imaginary parts, arranged in the order of ascending~$\Im(\kappa_n)$. Concerning the eigenvalues~$\kappa_n\!\in\mathbb{R}$, in this subset we include only those with $\kappa_n>0$, also arranged in the ascending order. 

\end{itemize}

To distinguish between the first and the second half of the set, we conveniently write $\kappa_n:=\leftw{\kappa}_{\!n}$ ($n=\overline{1,N}$) which refers to either propagation or exponential decay of $e^{i \kappa_n\es\xi_1}$ in the \emph{negative} $\xi_1$-direction (the so-called \emph{left-going} waves), and $\kappa_n:=\rightw{\kappa}_{\!n}$ ($n=\overline{N\!+\!1,2N}$) which refers to either propagation or exponential decay of $e^{i\kappa_n\es\xi_1}$ in the \emph{positive} $\xi_1$-direction (the so-called \emph{right-going} waves); see Fig.~\ref{trunc} for illustration. 

\begin{figure}[ht]
\centering{\includegraphics [width=0.62\textwidth ]{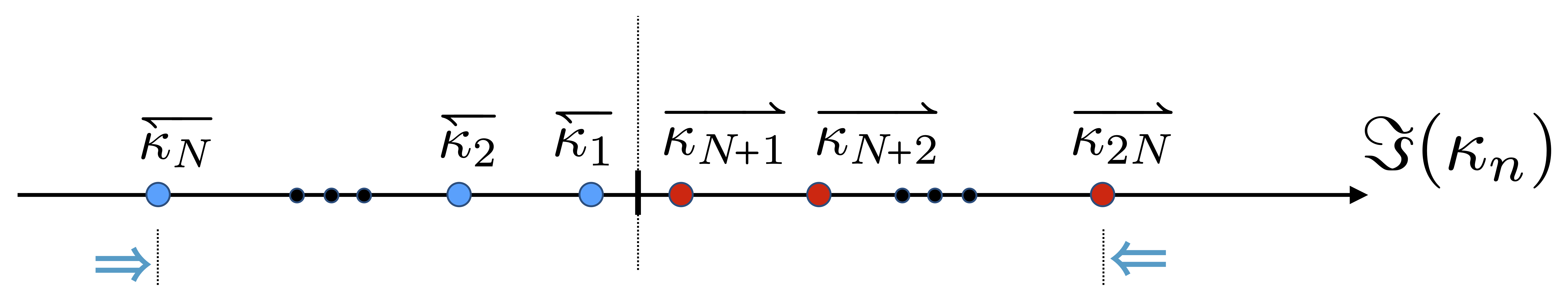}}
\caption{Truncated eigenspectrum of QEVP~\eqref{QEP} assuming $|\Im(\kappa_{n})|\!>\!0$, $n\shh \overline{1,2N}$.} \label{trunc}
\end{figure}

\begin{remark}
The foregoing truncation of the eigenspectrum of~\eqref{QEP} amounts to discarding the Bloch-wave modes $\phi_n(\bxi)e^{i (\kappa_n,k_2) \cdot \bxi}$ with a strong exponential decay in either $-\xi_1$ or $+\xi_1$ direction. Physically, these modes represent (either pure or leaky) evanescent waves~\citep[e.g.][]{laude2009evanA} propagating in the~$\xi_2$-direction; in the case of scattering by a layer-cake system that is periodic in the~$\xi_2$-direction (the focus of this study), such modes signify interfacial waves that are ``trapped''  by the system. In the context of the scattering problem examined in Sec.~\ref{scatbvp}, a justification of the featured truncation paradigm is provided in Sec.~\ref{strongeva}.
\end{remark}

\begin{remark}
For unit cells that are symmetric with respect to the $\xi_1\shh \ell/2$ axis, the symmetry argument guarantees that when~$\kappa_n\!\in\mathbb{R}$ is an eigenvalue of QEVP~\eqref{QEP}, so is~$-\kappa_n$. In such situations, one finds that (i) $\kappa_{m+N}=-\kappa_m$ for $\Im(\kappa_m)\shh 0$, and (ii) $\Im(\kappa_{m+N})=-\Im(\kappa_m)$ for $\Im(\kappa_m)\!>\!0$, $m=\overline{1,N}$. The foregoing arrangement, however, does not rely on such parity. When selecting~$N$, it is only important that both $\exp(\Im(\rightw\kappa_{\!\!N})\ell)$ and $\exp(-\Im(\leftw{\kappa}_{\!\!2N})\ell)$ fall below a selected threshold value, see Sec.~\ref{strongeva}. 
\end{remark}

%-----------------------------------------------------------------------------------------------------------------%
\subsection{Power flow of the propagating modes} \label{Poynt}
%----------------------------------------------------------------------------------------------------------------------%

As suggested above, it is useful to classify the Bloch waves $\phi_n(\bxi)e^{i (\kappa_n,k_2) \cdot \bxi}$ as either left- or right-going with respect to the $\xi_1$-direction. For evanescent modes, such classification is unequivocally decided by the sign of~$\Im(\kappa_n)$. For propagating modes ($\Im(\kappa_n)\shh 0$), on the other hand, the situation is more nuanced. In particular, we opt to classify the latter according to the sign of~$\kappa_n$, which identifies the direction of  phase delay. While such choice suffices for the present study, for propagating modes it may also be useful to identify the direction of \emph{energy propagation}. For lossless media (assumed in this study), the sought direction is identified via the group velocity, $\bv^{g} = \nabla_{\!\bk}\es \omega$. To expose the direction of~$\bv^{g}$, we recall the definition of the Poynting vector \citep{Auld1973} -- averaged temporally over the period of oscillations $\text{\sc{T}}\!=2\pi/\omega$ and spatially over~$Y$ -- which can be written as  
\begin{equation} \label{PV1}
\langx\boldsymbol{\mathcal{P}}_{\!\!\text{\tiny T}}\rangx \::=\: \frac{1}{2} \Re\hh \langx\! -\nn\overline{v}\es \bsig \rangx 
\:=\:  \frac{\omega}{2} \Im\hh \langx \overline{u}\es \bsig \rangx. 
\end{equation}
Here $\bsig = G(\bxi)\nabla u$ is the stress vector; $v=-i \omega u$ is the particle velocity, and $\langx\!\cdot\!\rangx = |Y|^{-1}\langle\cdot,1\rangle$ denotes the spatial average over $Y$. It is shown by Willis~\citep{willis2016} that for a Floquet-Bloch wave, $\langx\boldsymbol{\mathcal{P}}_{\!\!\text{\tiny T}}\rangx = \langx\mathcal{U}_{\text{\tiny T}}\rangx^{-1}\bv^{g}$ where $\langx\mathcal{U}_{\text{\tiny T}}\rangx$ is the spatiotemporal average of the internal  energy density. Typically, $\langx\boldsymbol{\mathcal{P}}_{\!\!\text{\tiny T}}\rangx$ is  easier to evaluate than~$\bv^{g}$ and can be conveniently used to expose its direction. Using \eqref{FBT}, we obtain $\bsig = G [(\nabla\nes + i\bk) \phi] e^{i\bk \cdot \bxi}$ and so 
\begin{equation} \label{PV2}
\langx\boldsymbol{\mathcal{P}}_{\!\!\text{\tiny T}}\rangx  \:=\:  
\frac{\omega}{2} \Im\hh \langx G \es\overline{\phi} \es  (\nabla\nes + i\bk) \phi \rangx
\end{equation}
for a given Bloch wave. In the context of this study, Floquet-Bloch modes with $\langx\boldsymbol{\mathcal{P}}_{\!\!\text{\tiny T}}\rangx\cdot\be_1\!>\!0$ carry energy to the right and vice versa. Thus, a more meaningful ``left-right''  taxonomy of the eigenvalues $\kappa_n\!\in\mathbb{R}$ could be established by considering the sign of $\langx\boldsymbol{\mathcal{P}}_{\!\!\text{\tiny T}}\rangx\sip\be_1$ in lieu of that of $k_1\shh \kappa_n$ (which signifies the direction of \emph{phase} propagation). For its simplicity, however, we retain the latter criterion for the classification of propagating Bloch modes.

%----------------------------------------------------------------------------------------------------------------------%
\section{Scattering problem} \label{scatbvp}
%----------------------------------------------------------------------------------------------------------------------%

We consider the scattering of scalar waves in an otherwise homogeneous medium by a layer-cake system occupying the region
\[
\Omega \,=\, \left\{\bxi\nes:\hh 0 < \xi_1\! < L, \: \xi_2\!\in\mathbb{R} \right\}, 
\]
that consists of $J$  dissimilar strips that are $d$-periodic in the $\xi_2$-direction, see Fig.~\ref{RT}. 

The mother periodic medium, $\mathcal{M}_j$, of the $j$th strip is supported on an orthogonal lattice with the unit cell 
\begin{equation} \label{motherY}
Y_j = \{\bxi\nes:\hh x^{(j)}\nes < \xi_1\nes < x^{(j)}\!+\nes\ell^{(j)}, \: 0 < \xi_2 < d \}, \quad j=\overline{1,J}.    
\end{equation}
In general, we allow that different strips be cutouts from the same periodic medium, i.e. that $\mathcal{M}_j\shh \mathcal{M}_l$ for $j\!\!\neq\!l$ and $j,l\in\overline{1,J}$; we shall make that distinction later in Sec.~\ref{plug-and-play}. The $j$th strip has width~${w}^{(j)}$ and occupies the region 
\[
\{\bxi\nes:\hh x^{(j)}\nes<\xi_1\nes <x^{(j+1)}, \: \xi_2\nes\in\mathbb{R}\}, \quad x^{(j+1)} =\, x^{(j)} + {w}^{(j)}, \quad j=\overline{1,J} 
\]
where $x^{(1)}\shh 0$ and $x^{(J+1)}\shh L$. For generality, we let ${w}^{(j)}\!\gtreqlessslant\nes \ell^{(j)}$, i.e. we do not require the strips to be unit-cell-wide.

To help formulate the scattering problem, we denote the strip boundaries by 
\begin{equation}
 \Gamma_{\!j}=\{\bxi\nes:\hh \xi_1 = x^{(j)}\}, \quad j=\overline{1,J\!+\!1}.    
\end{equation}
Across every~$\Gamma_{\!j}$, the wavefield~$u(\bxi)$ maintains the continuity of displacements and shear stresses in that 
\begin{equation} \label{continuity}
u\es|_{\Gamma_{\!j}^-} \:=\: u\es|_{\Gamma_{\!j}^+}, \qquad \sigma\es|_{\Gamma_{\!j}^-} \:=\: \sigma\es|_{\Gamma_{\!j}^+}, \quad \sigma= (1,0)\cdot G(\bxi) \nabla u, \quad j=\overline{1,J\!+\!1}. 
\end{equation}

Letting~$G_\circ$, $\rho_\circ$ and $c_\circ=\sqrt{G_\circ/\rho_\circ}$ denote respectively the shear modulus, mass density, and phase velocity of the homogeneous medium, our focus is on situations where the incident field 
\begin{equation}\label{uinc}
u^{inc} \,=\, e^{i \bk\cdot\bxi}, \qquad \bk=(k_1,k_2)\in\mathbb{R}^2, \quad \|\bk\|=\frac{\omega}{c_\circ}, \quad k_1>0  
\end{equation}
is a plane wave impinging (at frequency~$\omega$) on the layered system from the left. In order  to satisfy the germane continuity conditions across all interfaces parallel to the $\xi_2$-axis, we know from the onset that the $k_2$-component of the incident wave vector is a \emph{conserved quantity}~\citep{lamacz2018outgoing}. Indeed, this fact motivates consideration of the QEVP~\eqref{QEP}. 

\subsection{Preliminaries}
In what follows, quantities affiliated with the left half-space, strip $j\nes\in\nes\overline{1,J}$, and the right half-space will be indicated by the superscript $0$, $j$, and~$J\!+\!1$, respectively. For instance, the eigenvalues and eigenfunctions of the quadratic eigenvalue problem~\eqref{QEP} written for the $j$th unit cell will be denoted by~$\kappa_n^j$ and~$\phi_n^j$. By analogy to the earlier classification of the wavenumbers~$\kappa_n$, all quantities corresponding to either left-propagating or left-decaying waves will be indicated by the ``$\leftw{}$'' symbol, and similarly by ``$\rightw{}$'' for the right-going waves. In this vein, we start by defining the matrix exponential functions
\begin{equation} \label{aux1}
\bE^j(\xi_1) = \begin{bmatrix} \leftw{\bE}^j(\xi_1) & \bzero \\ \bzero & \rightw{\bE}^j(\xi_1) \end{bmatrix}, \qquad 
\begin{aligned} 
\leftw{\bE}^j(\xi_1) &\,=\, \text{diag}\,( e^{i \leftw{\kappa}^j_{\!n} \xi_1}, n=\overline{1,N} ), \\
\rightw{\bE}^j(\xi_1) &\,=\, \text{diag}\, ( e^{i \rightw{\kappa}^j_{\!n} \xi_1}, n=\overline{N+1,2N}  )
\end{aligned}. 
\end{equation}
Since~\eqref{aux1} features the wave numbers with both negative and positive imaginary parts, matrix $\bE^j$ carries a mix of exponentially-growing and exponentially-decaying functions. Thus, even for moderate values of~$\xi_1$, $\bE^j$ may be burdened by a large condition number and thus ill-suited for computational treatment of the scattering problem. To help regularize the problem without additional approximations, we next introduce the constant matrices 
\begin{equation} 
\bH ^j = \begin{bmatrix} \leftw{\bH }^j & \bzero \\ \bzero & \rightw{\bH }^j \end{bmatrix}, \qquad 
\begin{aligned}
\leftw{\bH }^j &\,=\, \text{diag}\,( e^{-i \leftw{\kappa}^j_{\!n}x^{(j+1)}}, n=\overline{1,N} ) \\
\rightw{\bH }^j &\,=\, \text{diag}\,( e^{-i \rightw{\kappa}^j_{\!n}x^{(j)}}, n=\overline{N+1,2N} ) 
\end{aligned},
\end{equation}
which allows us to bring to bear the regularized exponential matrices 
\begin{equation} \label{regexp}
\Ereg^j(\xi_1) = \bE^j(\xi_1) \es \bH^j = \begin{bmatrix}
\leftw{\Ereg}^j(\xi_1) & \bzero \\
\bzero & \rightw{\Ereg}^j(\xi_1)
\end{bmatrix}, \qquad 
\begin{aligned}
 \leftw{\Ereg}^j(\xi_1) &= \leftw{\bE}^j(\xi_1) \leftw{\bH }^j \\ \rightw{\Ereg}^j(\xi_1) &= \rightw{\bE}^j(\xi_1) \rightw{\bH }^j   
\end{aligned}.
\end{equation}
Considering the behavior of~$\Ereg^j$ for  $x^{(j)}\!<\xi_1\!<x^{(j+1)}$ and letting~$\mathcal{E}^j_n$ denote its $n$th diagonal entry ($n=\overline{1,2N}$), \eqref{regexp} ensures that (i) $|\mathcal{E}^j_n|\leqslant 1$, and (ii) $|\mathcal{E}^j_n|$ does not become vanishingly small even for large~$x^{(j)}$, e.g. in the case of a large number of unit cells. In the sequel, we shall use the short-hand notation 
\begin{equation} \label{Eregjq}
\rightw{\Ereg}^j_{\!\!q} := \rightw{\Ereg}^j(x^{(q)}),  \qquad  \leftw{\Ereg}^j_{\!\!q} := \leftw{\Ereg}^j(x^{(q)}) \qquad 
\text{for}~~ q\in\{j,j\!+\!1\},
\end{equation}
noting in particular that  
\begin{equation} \label{expomat1}
\rightw{\Ereg}^{j}_{\!\!j}=\leftw{\Ereg}^{j}_{\!\!j+1}=\bI, \quad~ 
\rightw{\Ereg}^{j}_{\!\!j+1} = \text{diag}\,( e^{i \rightw{\kappa}^j_{\!n}{w}^{(j)}}, n=\overline{1,N} ), \quad~ 
\leftw{\Ereg}^{j}_{\!\!j} = \text{diag}\,( e^{-i \leftw{\kappa}^j_{\!n}{w}^{(j)}}, n=\overline{1,N} ), \quad~ j=\overline{1,J}
\end{equation}
with $\Im(\leftw{\kappa}^j_{\!n})\leqslant 0$ and $\Im(\rightw{\kappa}^j_{\!n})\geqslant 0$.

\begin{figure}
\centering{\includegraphics [width=0.9\textwidth ]{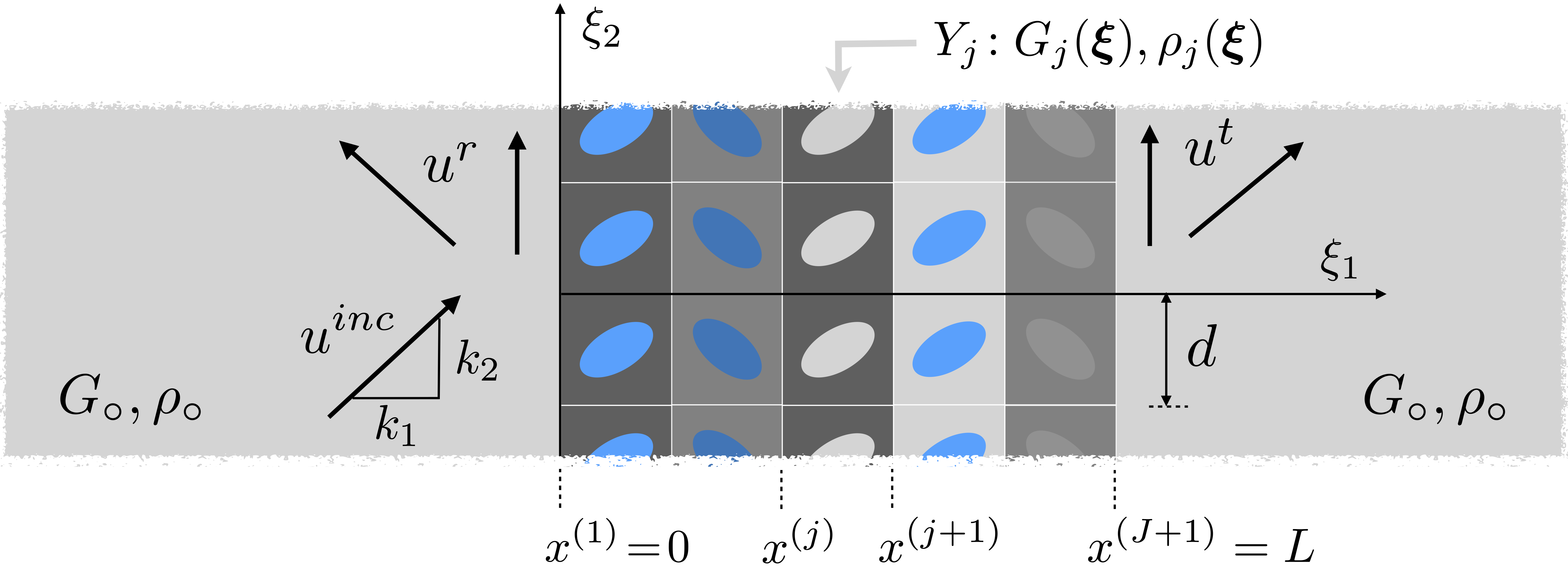}}
\caption{Scattering of scalar waves by a layer-cake system~$\Omega$ that is $d$-periodic in the $\xi_2$-direction. For clarity, each strip is depicted as being unit-cell-wide; this is not required by the analysis, see Fig.~\ref{strips} for a generic configuration.} \label{RT}
\end{figure}

%----------------------------------------------------------------------------------------------------------------------%
\subsection{Wavefield in the layer-cake system}\label{SecLCS}
%----------------------------------------------------------------------------------------------------------------------%

On recalling the solution of the QEVP~\eqref{QEP} and its truncated spectrum $\{\kappa_m,\phi_m\}_{m=1}^{2N}$ introduced in Sec.~\ref{trunceigsp}, the displacement field~$u^j$ and  affiliated shear stress $\sigma^j=(1,0)\cdot G_{\nes j} \nabla u^j$ on a ``vertical'' plane within the $j$th strip of the layer-cake system can be approximated as 
\begin{equation}\label{fields}
\begin{aligned}
u^j(\bxi) &= \sum_{n=1}^{2N} \alpha^j_n \phi^j_n(\bxi) e^{{i} \bk_n^j \cdot \bxi}, \quad \bk_n^j = (\kappa_n^j,k_2) \\
\sigma^j(\bxi) &=  \sum_{n=1}^{2N} \alpha^j_n \psi^j_n(\bxi)e^{{i} \bk_n^j \cdot \bxi},  \quad \psi^j_n(\bxi) = G_{\nes j}\nes(\bxi) \left(\frac{\partial}{\partial \xi_1} + i \kappa_n^j\right) \phi^j_n(\bxi). 
\end{aligned}
\end{equation}
In physical terms, approximation~\eqref{fields} discards the Floquet-Bloch wave modes with a rapid exponential decay in either~$\xi_1$ or~$-\xi_1$ direction. 

\begin{remark}
To cater for strips that are wider than the unit cell $Y_j$ of the mother periodic medium $\mathcal{M}_j$ (i.e. ${w}^{(j)}\nes>\ell^{(j)}\nes$), $\phi^j_n, \, \psi^j_n$ and $G_{\nes j}$ in~\eqref{fields} are hereon interpreted as the $Y_j$-periodic maps $\mathbb{R^2}\nes\mapsto\nes\mathbb{C}$.
\end{remark}

By virtue of~\eqref{2DFB}, the periodicity of~$\phi^j_n(\bxi)$ and~$\psi^j_n(\bxi)$ in the $\xi_2$-direction is~$d$, while that of the factor~$e^{{i} \bk_n^j \cdot \bxi}$ is $2\pi/k_2$. As a result, $u^j$ will exhibit periodicity in the $\xi_2$-direction only in situations where $k_2 d/(2\pi)$ is a rational number. To handle the problem, we conveniently proceed by solving for the factored \citep{Kulpe2014,Sriva2020} displacement and (shear) stress fields
\begin{equation}\label{factusig}
\tilde{u}^j(\bxi) :=\, u^j(\bxi) \, e^{-i k_2 \xi_2}, \qquad~~
\tilde{\sigma}^j(\bxi) :=\, \sigma^j(\bxi) \, e^{-i k_2 \xi_2}
\end{equation}
that are $d$-periodic in the $\xi_2$-direction. As a result, one can expand the quantities in~\eqref{factusig} in Fourier series, whose truncated version reads  
\begin{equation} \label{fields2}
\begin{aligned}
\tilde{u}^j(\bxi) &= 
\sum_{m=-M}^{M} \left\{ \sum_{n=1}^{2N} \alpha_n^j \hat{\Phi}_{mn}^j(\xi_1) e^{{i} \kappa_n^j \xi_1} \right\} e^{{i} \frac{2 \pi}{d} m \xi_2 }, \qquad 
\hat{\Phi}_{mn}^j(\xi_1) = \frac{1}{d} \int_{0}^{d} \phi^j_n(\bxi) e^{-{i} \frac{2 \pi}{d} m \xi_2} d\xi_2, \\ 
\tilde{\sigma}^j(\bxi) &= 
\sum_{m=-M}^{M} \left\{ \sum_{n=1}^{2N} \alpha_n^j \hat{\Psi}^j_{mn}(\xi_1) e^{{i} \kappa_n^j \xi_1} \right\} e^{{i} \frac{2 \pi}{d} m \xi_2 }, \qquad
\hat{\Psi}^j_{mn}(\xi_1) = \frac{1}{d} \int_{0}^{d} \psi^j_n(\bxi) e^{-{i} \frac{2 \pi}{d} m \xi_2} d\xi_2.
\end{aligned}
\end{equation}
Letting    
\begin{equation} \label{fields2z}
\begin{aligned}
\hat{\tilde{\bu}}^j(\xi_1) := 
[\hat{\tilde{u}}^j_m, m=\overline{-M,M}]\tran, \qquad \hat{\tilde{u}}^j_{m}(\xi_1) = \sum_{n=1}^{2N} \alpha_n^j \hat{\Phi}_{mn}^j(\xi_1) e^{{i} \kappa_n^j \xi_1} \\
\hat{\tilde{\bsig}}^j(\xi_1) := 
[\hat{\tilde{\sigma}}^j_{m}, m=\overline{-M,M}]\tran, \qquad \hat{\tilde{\sigma}}^j_{m}(\xi_1) = \sum_{n=1}^{2N} \alpha_n^j \hat{\Psi}_{mn}^j(\xi_1) e^{{i} \kappa_n^j \xi_1}
\end{aligned}
\end{equation}
collect the respective Fourier components of $\tilde{u}^j(\bxi)$ and $\tilde{\sigma}^j(\bxi)$, using matrix-vector notation we obtain 
\begin{equation} \label{eqRT}
\begin{bmatrix} \hat{\tilde{\bu}}^j \\ \hat{\tilde{\bsig}}^j \end{bmatrix} (\xi_1) = 
\bLam ^j(\xi_1) \; \bE^j(\xi_1) \, \boldsymbol{\alpha}^j, \qquad j=\overline{1,J}
\end{equation}
where~$\boldsymbol{\alpha}^j$ collects~$\alpha_n^j$; $\,\bE^j$ is given by~\eqref{aux1}, and 
\begin{equation}\label{LamLam}
\bLam^j(\xi_1) \:=\: \begin{bmatrix}
\hat{\Phi}_{(-M)n}^j, & n=\overline{1,2N}  \\ 
\vdots & \vdots\\
\hat{\Phi}_{Mn}^j, &n=\overline{1,2N}  \\
\hat{\Psi}_{(-M)n}^j, &n=\overline{1,2N} \\
\vdots & \vdots \\
\hat{\Psi}_{Mn}^j, &n=\overline{1,2N}
\end{bmatrix} \: := \: 
\begin{bmatrix}
\leftw{\bLam}^j(\xi_1) & 
\rightw{\bLam}^j(\xi_1)
\end{bmatrix} 
\end{equation}
which differentiates between the coefficients affiliated with the ``left-going'' waves (columns~1 through~$N$) and those affiliated with the ``right-going'' waves (columns~$N\!+\!1$ through~$2N$). In the spirit of~\eqref{Eregjq}, we reserve the short-hand notation 
\begin{equation} \label{Lamjq}
\rightw{\bLam}^j_{q} \,:=\, \rightw{\bLam}^j(x^{(q)}),  \qquad  \leftw{\bLam}^j_{q} \,:=\, \leftw{\bLam}^j(x^{(q)}) \qquad
\text{for}~~ q\in\{j,j\!+\!1\}
\end{equation}
for future use. 

\begin{remark} \label{MN-relat}
We observe that the dimension of~$\bLam^j$ is $(4M\! +\nes 2)\times(2N)$, as controlled by the number of Fourier modes $(2M\!+\nes1)$ and the number of Bloch eigenstates~$(2N)$ involved. As a rule of thumb, one seeks to have $4M\!+\nes 2\geqslant 2N$ which yields either an ``even-determined'' or ``overdetermined'' system where the number of continuity conditions (written in the Fourier space) is equal or larger than the number of unknowns. At this point, however, little is known about the rank of~$\bLam^j$, even though numerical simulations suggest that it equals~$2N$.   
\end{remark}

On recalling~\eqref{regexp}, we next rewrite~\eqref{eqRT} in a regularized form 
\begin{equation} \label{fourcoefint}
\begin{aligned}
\begin{bmatrix} \hat{\tilde{\bu}}^j \\ \hat{\tilde{\bsig}}^j \end{bmatrix} (\xi_1) &= 
\bLam ^j(\xi_1) \: \Ereg^j(\xi_1) \: \bbeta  ^j \\*[-2mm]
&= [\leftw{\bLam}^j(\xi_1) \quad \rightw{\bLam}^j(\xi_1)] \; \Ereg^j(\xi_1) 
\begin{bmatrix} \leftw{\bbeta}^{\nes j} \\ \rightw{\bbeta}^{\nes j} \end{bmatrix}, \qquad j=\overline{1,J}
\end{aligned}
\end{equation}
that features a new set of constants
\begin{equation} \label{betas}
\bbeta^j = \begin{bmatrix} (\leftw{\bH}^j)^{-1} & \bzero \\ \bzero & (\rightw{\bH}^j)^{-1} \end{bmatrix} 
\boldsymbol{\alpha}^j \: := \: \begin{bmatrix} \leftw{\bbeta}^{\nes j} \\ \rightw{\bbeta}^{\nes j} \end{bmatrix},
\end{equation}
again parsed by the direction of propagation or decay. 

%----------------------------------------------------------------------------------------------------------------------%
\subsubsection{Translations of the mother periodic media} \label{sectran}
%----------------------------------------------------------------------------------------------------------------------%

With reference to~\eqref{motherY}, we next consider a \emph{translation} of the mother periodic medium~$\mathcal{M}_j$ by vector $\bs_{\nn j} \in[0,\ell^{(j)})\times[0,d)$ relative to~$Y_j$, see Fig.~\ref{translation}. Such alteration, which inherently affects the impedance mismatches -- and thus energy and phase transfer -- between the laminae, is motivated by the possibility of exploring the design space of a layer-cake system for the purposes of e.g. rainbow trapping or energy harvesting. On denoting the resulting field quantities by the prime symbol, we obtain 
\begin{equation}
G_{\nes j}'(\bxi) \,=\, G_{\nes j}(\bxi-\bs_{\nn j}), \qquad \rho_{\nes j}'(\bxi) \,=\, \rho_{\nes j}(\bxi-\bs_{\nn j}).
\end{equation}
In this setting, \eqref{fields} becomes 
\begin{equation}\label{fields2bis}
u^{j\prime}(\bxi) = \sum_{n=1}^{2N} \alpha^j_n \phi^j_n(\bxi-\bs_{\nn j}) e^{{i} \bk_n^j \cdot \bxi}, \qquad 
\sigma^{j\prime}(\bxi) =  \sum_{n=1}^{2N} \alpha^j_n \psi^j_n(\bxi-\bs_{\nn j})e^{{i} \bk_n^j \cdot \bxi},  \qquad 
\quad \bk_n^j = (\kappa_n^j,k_2)
\end{equation}
where $\kappa_n^j$, $\phi^j_n(\bxi)$ and $\psi^j_n(\bxi)$ are the eigenvalue and eigenfunctions featured in~\eqref{fields}. As result, the  exponential (eigenvalue) matrices~\eqref{regexp} remain unchanged, i.e. 
\[
\rightw{\Ereg}^{j\prime}(\xi_1)= \rightw{\Ereg}^j(\xi_1), \qquad \leftw{\Ereg}^{j\prime}(\xi_1)= \leftw{\Ereg}^j(\xi_1), 
\]
while the eigenfunction matrices $\leftw{\bLam}^{j\prime}(\xi_1)$ and $\rightw{\bLam}^{j\prime}(\xi_1)$ must be recomputed via~\eqref{fields2}--\eqref{LamLam} by deploying the shifted eigenfunctions $\phi^j_n(\bxi-\bs_{\nn j})$ and $\psi^j_n(\bxi-\bs_{\nn j})$ in lieu of $\phi^j_n(\bxi)$ and $\psi^j_n(\bxi)$, respectively. Any such reevaluation, however, entails only a minimal computational cost for it voids the need for solving anew QEVP~\eqref{QEP} for each~$\mathcal{M}_j$. 

In this way, by letting $\bs_{\nn j}\shh \mu\hh\be_2$ for some $\mu\nes\in\nes (0,d)$ (``vertical'' translation) one may for instance generate strips of periodic media supported on \emph{non-orthogonal} Bravais lattices, see the last three strips in Fig.~\ref{strips}(c) for illustration. In such situations, $\leftw{\bLam}^{j\prime}(\xi_1)$ and $\rightw{\bLam}^{j\prime}(\xi_1)$ can be quickly recomputed via the shift property of the Fourier series in that $\hat{\Phi}_{mn}^{j\prime}(\xi_1)=\hat{\Phi}_{mn}^j(\xi_1) \, e^{-{i} \frac{2 \pi}{d} m \mu}$ and similarly for $\hat{\Psi}_{mn}^{j\prime}(\xi_1)$.

\begin{figure}
\centering{\includegraphics [width=0.84\textwidth ]{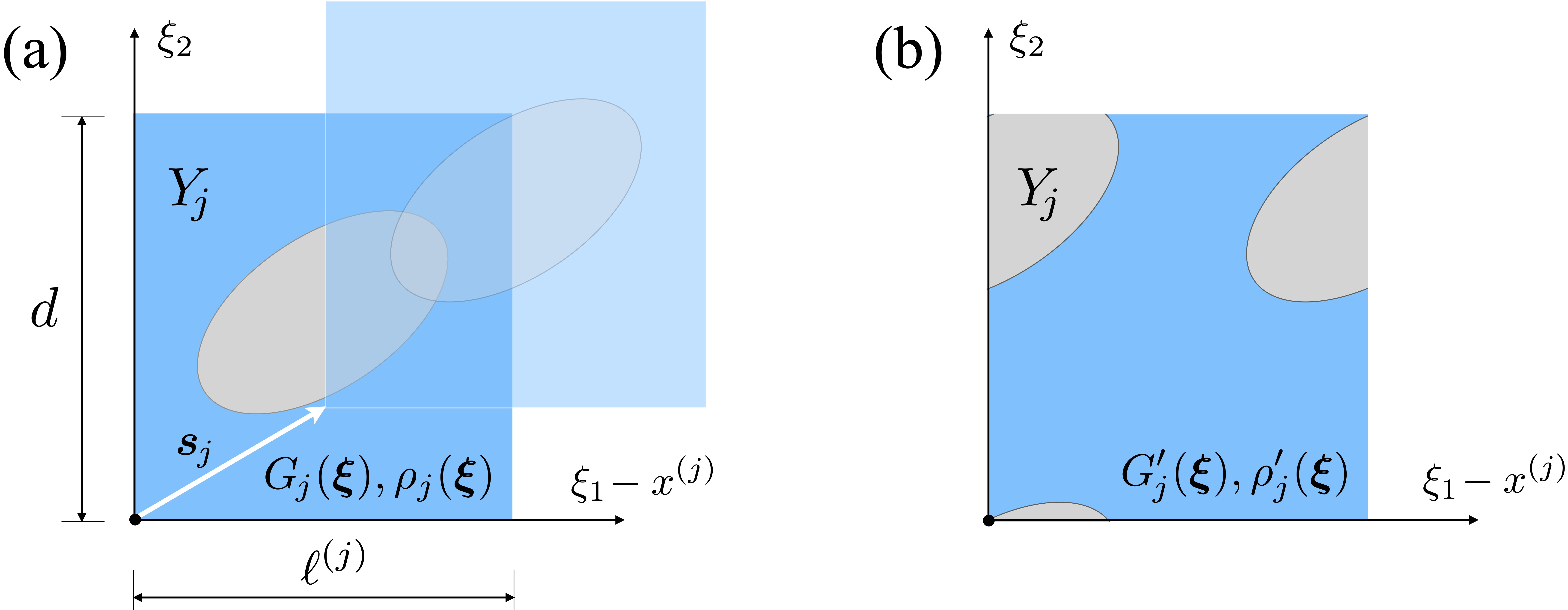}} 
\caption{(a) Translation by vector~$\bs_{\nn j}$ of the mother periodic medium~$\mathcal{M}_j$ relative to a fixed window~\eqref{motherY} of the unit cell~$Y_j$, and (b) the resulting unit cell morphology of the $j$th strip.}\label{translation}
\end{figure}

%----------------------------------------------------------------------------------------------------------------------%
\subsection{Wavefield in the homogeneous medium}
%----------------------------------------------------------------------------------------------------------------------%

On denoting respectively by~$u^r$ and~$u^t$ the wavefield in the homogeneous medium reflected and transmitted by the heterogeneous layer, we proceed by computing the factored displacement and stress fields in the left ($-$) and the right ($+$) homogeneous half-space given by
\begin{equation} \notag 
\begin{aligned}
\tilde{u}^{-} &= (u^{inc}+u^r)\, e^{-i k_2 \xi_2}, \qquad &\tilde{\sigma}^{-} &= (\sigma^{inc}+\sigma^r)\, e^{-i k_2 \xi_2}  \\
\tilde{u}^{+} &= u^t\, e^{-i k_2 \xi_2}, \qquad &\tilde{\sigma}^{+} &= \sigma^t\, e^{-i k_2 \xi_2}.
\end{aligned}
\end{equation}
Since~$\tilde{u}^r=u^r(\bxi)e^{-i k_2 \xi_2}$, $\tilde{u}^t=u^t(\bxi)e^{-i k_2 \xi_2}$ and their stress counterparts are $d$-periodic in the $\xi_2 $-direction, one can again make use of the (truncated) Fourier series expansion in that $\tilde{u}^r=\sum_{m=-M}^{M}\hat{\tilde{u}}^r_m(\xi_1) e^{{i} \frac{2 \pi}{d} m \xi_2}$ and similarly for the other fields. Since both half-spaces are homogeneous, each Fourier component of these factored fields is a \emph{plane wave} propagating in the~$\xi_1$-direction, e.g. $\,\hat{\tilde{u}}^r_m(\xi_1) = r_m e^{{i} k_{1m}^- \xi_1}$ for some constant~$r_m$ and wavenumber~$k_{1m}^-$. Letting  
\begin{equation} \label{imped1}
 S_\circ \,=\, i G_\circ k_1, \qquad S_m^{\pm} \,=\, i G_\circ k_{1m}^{\pm},
\end{equation}
one can write 
\begin{equation} \label{usigmp1} 
\begin{aligned}
\tilde{u}^-(\bxi) &= e^{{i} k_1 \xi_1} +\! \sum_{m=-M}^{M} (r_m e^{{i} k_{1m}^- \xi_1}) e^{{i} \frac{2 \pi}{d} m \xi_2}, \quad
&\tilde{\sigma}^{-}(\bxi) &= S_\circ e^{{i} k_1 \xi_1} +\! \sum_{m=-M}^{M} (S_m^-\hh r_m e^{{i} k^-_{1m} \xi_1}) e^{{i} \frac{2 \pi}{d} m \xi_2},    \\
\tilde{u}^+(\bzeta) &= \sum_{m=-M}^{M} (t_m e^{{i} k_{1m}^+ \zeta_1}) e^{{i} \frac{2 \pi}{d} m \zeta_2}
&\tilde{\sigma}^{+}(\bzeta) &= \sum_{m=-M}^{M} (S_m^+\hh t_m e^{{i} k_{1m}^+ \zeta_1}) e^{{i} \frac{2 \pi}{d} m \zeta_2},  
\end{aligned}
\end{equation}
where the transmitted wavefield is conveniently described in terms of the shifted coordinate \mbox{$\bzeta = \bxi\!-\!(L,0)$} whose origin resides on~$\Gamma_{\!J+1}$. Since $\bk^-_m\shh (k_{1m}^-,k_2+2\pi m/d)$ is the wave vector of the component of~$u^r$ due to~$\hat{\tilde{u}}^r_m$, from the requirement that $\|\bk^-_m\|=\omega/c_\circ$ (see also~\eqref{uinc}) we obtain 
\begin{equation} \label{homogks}
k_{1m}^\pm \:=\: \pm \sqrt{\left(\frac{\omega}{c_\circ}\right)^2 - \left( k_2 + \frac{2 \pi}{d} m \right)^2} \quad \Rightarrow \quad k^\pm_ {10} = \pm k_1, \quad  S_m^+ = - S_m^-
\end{equation}
making use of the analogous argument in terms of $\bk^+_m$ and the fact that $k_{1m}^-\leqslant 0$ (reflected wave) whereas $k_{1m}^+\geqslant 0$ (transmitted  wave). On the basis of~\eqref{usigmp1}, we can formally write 
\begin{equation} \label{usigmp2} 
\begin{aligned}
\tilde{u}^-(\bxi) &:= \sum_{m=-M}^{M} \hat{\tilde{u}}^{0}_m(\xi_1) \es e^{{i} \frac{2 \pi}{d} m \xi_2}, \quad
&\tilde{\sigma}^{-}(\bxi) &:= \sum_{m=-M}^{M} \hat{\tilde{\sigma}}^{0}_m(\xi_1) \es e^{{i} \frac{2 \pi}{d} m \xi_2},  \\
\tilde{u}^+(\bzeta) &:= \sum_{m=-M}^{M} \hat{\tilde{u}}^{J+1}_m(\zeta_1) \es e^{{i} \frac{2 \pi}{d} m \zeta_2},
&\tilde{\sigma}^{+}(\bzeta) &:= \sum_{m=-M}^{M} \hat{\tilde{\sigma}}^{J+1}_m(\zeta_1) \es e^{{i} \frac{2 \pi}{d} m \zeta_2},
\end{aligned}
\end{equation}
and 
\begin{equation} \label{fields2x}
\begin{aligned}
\hat{\tilde{\bu}}^0(\xi_1) &= [\hat{\tilde{u}}^{0}_m, m=\overline{-M,M}]\tran, \qquad 
&\hat{\tilde{\bsig}}^0(\xi_1) &= [\hat{\tilde{\sigma}}^{0}_{m}, m=\overline{-M,M}]\tran, \\*[1mm]
\hat{\tilde{\bu}}^{J+1}(\zeta_1) &= [\hat{\tilde{u}}^{J+1}_m, m=\overline{-M,M}]\tran, \qquad 
&\hat{\tilde{\bsig}}^{J+1}(\zeta_1) &= [\hat{\tilde{\sigma}}^{J+1}_{m}, m=\overline{-M,M}]\tran. 
\end{aligned}
\end{equation}
By virtue of~\eqref{usigmp1}--\eqref{usigmp2}, the vectors of Fourier coefficients in~\eqref{fields2x} can be compactly written as 
\begin{equation} \label{fourcoefext}
\begin{bmatrix} \hat{\tilde{\bu}}^0 \\ \hat{\tilde{\bsig}}^0 \end{bmatrix}(\xi_1) = 
 [\leftw{\bLam}^{\!\nes\circ} \quad \rightw{\bLam}^{\!\nes\circ} ] \, \Ereg^\circ(\xi_1) 
 \begin{bmatrix} \leftw{\br} \\ \rightw{\bi} \end{bmatrix}, \qquad~~
\begin{bmatrix} \hat{\tilde{\bu}}^{J+1} \\ \hat{\tilde{\bsig}}^{J+1} \end{bmatrix}(\zeta_1) = 
 [\leftw{\bLam}^{\!\nes\circ}  \quad \rightw{\bLam}^{\!\nes\circ} ] \, \Ereg^{\circ}(\zeta_1) \begin{bmatrix} \leftw{\bzero} \\ \rightw{\bt} \end{bmatrix},
\end{equation}
where 
\begin{equation} \label{aux0}
\Ereg^\circ(\xi_1) \:=\: 
\begin{bmatrix} \leftw{\Ereg}^{\!\circ}(\xi_1) & \bzero \\ \bzero & \rightw{\Ereg}^{\!\circ}(\xi_1) \end{bmatrix}, \qquad 
\begin{aligned}
\leftw{\Ereg}^{\!\circ}(\xi_1) &\,=\, \text{diag}\,( e^{i k_{1m}^-\xi_1}, m=\overline{-M,M}), \\
\rightw{\Ereg}^{\!\circ}(\xi_1) &\,=\, \leftw{\Ereg}^{\!\circ}(-\xi_1);
\end{aligned}
\end{equation}
vector~$\leftw{\br}$ (resp.~$\rightw{\bt}$) collects~$r_m$ (resp.~$t_m$) featured by~\eqref{usigmp1} for $m=\overline{-M,M}$; $\leftw{\bi}$ is the vector with entries $i_m = \delta_{m(M+1)}$ (with~$\delta$ being the Kronecker delta); $\leftw{\bzero}$ is a zero vector of length~$2M\!+\!1$, and 
\begin{equation}
[\leftw{\bLam}^{\!\nes\circ}  \quad \rightw{\bLam}^{\!\nes\circ}] = \begin{bmatrix} \bI & \bI \\ \boldsymbol{S} & -\boldsymbol{S} \end{bmatrix}, \qquad  
\boldsymbol{S} = \text{diag}\, \big( S^-_m, m=\overline{-M,M} \big)
\end{equation}
with $\bI$ being the identity matrix of rank $2M\!+1$. We note for future reference that~$\leftw{\Ereg}^{\!\circ}(0) = \rightw{\Ereg}^{\!\circ}(0)= \bI$.

%----------------------------------------------------------------------------------------------------------------------%
\section{Factorized Bloch solution} \label{FBS}
%----------------------------------------------------------------------------------------------------------------------%

We are now in position to impose the continuity conditions $e^{-i k_2\xi_2}\times\eqref{continuity}$ in terms of the Fourier coefficients~\eqref{fourcoefint} and~\eqref{fourcoefext}, synthesizing the wave motion in the heterogeneous layer and the homogeneous exterior, respectively. For clarity, we summarize the featured expressions as  
\begin{align}
\begin{aligned}
\begin{bmatrix} \hat{\tilde{\bu}}^j \\ \hat{\tilde{\bsig}}^j \end{bmatrix} (\xi_1) &= 
\big[\leftw{\bLam}^j(\xi_1) \quad \rightw{\bLam}^j(\xi_1)\big] \; 
\begin{bmatrix} \leftw{\Ereg}^j(\xi_1) & \bzero \\ \bzero & \rightw{\Ereg}^j(\xi_1) \end{bmatrix}
\begin{bmatrix} \leftw{\bbeta}^{\nes j} \\ \rightw{\bbeta}^{\nes j} \end{bmatrix}, \quad j=\overline{1,J} \\
\begin{bmatrix} \hat{\tilde{\bu}}^0 \\ \hat{\tilde{\bsig}}^0 \end{bmatrix}(\xi_1) &=
\big[\leftw{\bLam}^{\!\nes\circ} \quad \rightw{\bLam}^{\!\nes\circ} \big] \, 
\begin{bmatrix} \leftw{\Ereg}^{\!\circ}(\xi_1) & \bzero \\ \bzero & \rightw{\Ereg}^{\!\circ}(\xi_1) \end{bmatrix}
\begin{bmatrix} \leftw{\br} \\ \rightw{\bi} \end{bmatrix}, \\
\begin{bmatrix} \hat{\tilde{\bu}}^{J+1} \\ \hat{\tilde{\bsig}}^{J+1} \end{bmatrix}(\zeta_1) &= 
 \big[\leftw{\bLam}^{\!\nes\circ}  \quad \rightw{\bLam}^{\!\nes\circ} \big] \, 
 \begin{bmatrix} \leftw{\Ereg}^{\!\circ}(\zeta_1) & \bzero \\ \bzero & \rightw{\Ereg}^{\!\circ}(\zeta_1) \end{bmatrix}
 \begin{bmatrix} \leftw{\bzero} \\ \rightw{\bt} \end{bmatrix}.
\end{aligned} 
\end{align}
On expanding $e^{-i k_2\xi_2}\times\eqref{continuity}$ across~$\Gamma_1$ ($\xi_1\shh 0$) and~$\Gamma_{\!J+1}$ ($\xi_1\shh L$ i.e.~$\zeta_1\shh 0$) in Fourier series and deploying the orthonormality of~$e^{{i} \frac{2 \pi}{d} m \xi_2}$, we obtain 
\begin{equation} \label{sol1}
\begin{aligned}
\big[\leftw{\bLam}^{\!\nes\circ} \quad \rightw{\bLam}^{\!\nes\circ} \big] \, 
\begin{bmatrix} \leftw{\br} \\ \rightw{\bi} \end{bmatrix} &= 
\big[\leftw{\bLam}^{\!1}_{\!1} \quad \rightw{\bLam}^{\!1}_{\!1}\big] \; 
\begin{bmatrix} \leftw{\Ereg}^{\!1}_{\!1} & \bzero \\ \bzero & {\bmII} \end{bmatrix}
\begin{bmatrix} \leftw{\bbeta}^{\!1} \\ \rightw{\bbeta}^{\!1} \end{bmatrix} &\implies \quad 
\begin{bmatrix} \leftw{\br} \\ \rightw{\bbeta}^{\!1}\end{bmatrix} &= 
\begin{bmatrix} \leftw{\bT}_{\!1} & \rightw{\bR}_{1} \\ \leftw{\bR}_{1} & \rightw{\bT}_{\!1} \end{bmatrix} 
\begin{bmatrix} \leftw{\bbeta}^{\!1} \\ \rightw{\bi} \end{bmatrix} \\
 \big[\leftw{\bLam}^{\!\nes\circ}  \quad \rightw{\bLam}^{\!\nes\circ} \big] \, 
 \begin{bmatrix} \leftw{\bzero} \\ \rightw{\bt} \end{bmatrix} &= 
 \big[\leftw{\bLam}^{\!J}_{\!J+1} \quad \rightw{\bLam}^{\!J}_{\!J+1}\big] \; 
\begin{bmatrix} {\bmII}  & \bzero \\ \bzero & \rightw{\Ereg}^{\!J}_{\!\!J+1} \end{bmatrix}
\begin{bmatrix} \leftw{\bbeta}^{\!J} \\ \rightw{\bbeta}^{\!J} \end{bmatrix} &\implies \quad 
\begin{bmatrix} \leftw{\bbeta}^{\!J} \\ \rightw{\bt} \end{bmatrix} &= 
\begin{bmatrix} \leftw{\bT}_{\!\!J+1} & \rightw{\bR}_{J+1} \\ \leftw{\bR}_{J+1} & \rightw{\bT}_{\!\!J+1} \end{bmatrix} 
\begin{bmatrix} \leftw{\bzero} \\ \rightw{\bbeta}^{\!J} \end{bmatrix}
\end{aligned} \; ,
\end{equation}
where~$\bmII$ is the identity matrix of rank~$N$; $\bzero$ is a trivial rectangular matrix of suitable dimensions, and  
\begin{equation} \label{solx1}
\begin{aligned}
\begin{bmatrix} \leftw{\bT}_{\!1}  & \rightw{\bR}_{1} \\ \leftw{\bR}_{1}  &   \rightw{\bT}_{\!1} \end{bmatrix} &= 
\big[-\!\leftw{\bLam}^{\!\nes\circ} \quad \rightw{\bLam}^{\!1}_{\!1}\big]^{\dagger} \: \big[-\!\leftw{\bLam}^{\!1}_{\!1} \quad \rightw{\bLam}^{\!\nes\circ}\big] 
\begin{bmatrix} \leftw{\Ereg}^{\!1}_{\!\!1} & \bzero \\ \bzero & \bI \end{bmatrix}, \\
\begin{bmatrix} \leftw{\bT}_{\!\!J+1} & \rightw{\bR}_{J+1} \\ \leftw{\bR}_{J+1} & \rightw{\bT}_{\!\!J+1} \end{bmatrix} &= 
\big[-\!\leftw{\bLam}^{\!J}_{\!J+1} \quad \rightw{\bLam}^{\!\nes\circ}\big]^{\dagger} \: \big[-\!\leftw{\bLam}^{\!\nes\circ} \quad \rightw{\bLam}^{\!J}_{\!J+1}\big] 
\begin{bmatrix} \bI & \bzero \\ \bzero & \rightw{\Ereg}^{\!J}_{\!\!J+1} \end{bmatrix}, 
\end{aligned}
\end{equation}
where $\boldsymbol{B}^\dagger$ denotes the generalized (Moore-Penrose) inverse of~$\boldsymbol{B}$. Similarly, imposition of the continuity conditions across internal interfaces $\Gamma_{\!j}$ ($j=\overline{2,J}$) yields
\begin{multline} \label{sol2}
\big[\leftw{\bLam}^{\!j-1}_{\!j} \quad \rightw{\bLam}^{\!j-1}_{\!j}\big] 
\begin{bmatrix} {\bmII} & \bzero \\ \bzero & \rightw{\Ereg}^{\!j-1}_{\!\!j} \end{bmatrix} 
\begin{bmatrix} \leftw{\bbeta}^{\nes j-1} \\ \rightw{\bbeta}^{\nes j-1} \end{bmatrix} = 
\big[\leftw{\bLam}^{\!j}_{\!j} \quad \rightw{\bLam}^{\!j}_{\!j}\big] 
\begin{bmatrix} \leftw{\Ereg}^{\!j}_{\!\!j}  & \bzero \\ \bzero & {\bmII} \end{bmatrix} 
\begin{bmatrix} \leftw{\bbeta}^{\nes j} \\ \rightw{\bbeta}^{\nes j} \end{bmatrix} \\ \implies 
\begin{bmatrix} \leftw{\bbeta}^{\nes j-1} \\ \rightw{\bbeta}^{j} \end{bmatrix} = 
\begin{bmatrix} \leftw{\bT}_{\!\!j} & \rightw{\bR}_{j} \\ \leftw{\bR}_{j} & \rightw{\bT}_{\!\!j} \end{bmatrix} 
\begin{bmatrix} \leftw{\bbeta}^{\nes j} \\ \rightw{\bbeta}^{\nes j-1} \end{bmatrix},
\end{multline}
where 
\begin{equation} \label{solx3}    
\begin{bmatrix} \leftw{\bT}_{\!\!j} & \rightw{\bR}_{j} \\ \leftw{\bR}_{j} & \rightw{\bT}_{\!\!j} \end{bmatrix} = 
\big[-\!\leftw{\bLam}^{\!j-1}_{\!j} \quad \rightw{\bLam}^{\!j}_{\!j}\big]^{\dagger} \: \big[-\!\leftw{\bLam}^{\!j}_{\!j} \quad \rightw{\bLam}^{\!j-1}_{\!j}\big]
\begin{bmatrix} \leftw{\Ereg}^{\!j}_{\!\!j}  & \bzero \\ \bzero & \rightw{\Ereg}^{\!j-1}_{\!\!j} \end{bmatrix}, \qquad j=\overline{2,J}.
\end{equation}
Note that the dimension of the matrices being inverted in~\eqref{solx1} is $(4M\!+2)\times(2M\!+N\!+1)$, whereas that of the matrices  inverted in~\eqref{solx3} is $(4M\!+2)\times(2N)$. For completeness, we also list the dimensions of the featured reflection and transmission matrices:
\begin{equation} \notag
\begin{aligned}
\begin{aligned}
\leftw{\bT}_{\!1}, \: \rightw{\bT}_{\!\!J+1}: &~~ (2M\!+1)\times N  \quad  &~~ \rightw{\bR}_{1}, \: \leftw{\bR}_{J+1}: &~~ (2M\!+1)\times(2M\!+1)   \\
 \rightw{\bT}_{\!1}, \: \leftw{\bT}_{\!\!J+1}: &~~ N\times (2M\!+1)  \quad &~~ \leftw{\bR}_{1}, \: \rightw{\bR}_{J+1}: &~~ N\times N 
 \end{aligned} \\ 
 \leftw{\bT}_{\!\!j}, \: \rightw{\bT}_{\!\!j}, \: \leftw{\bR}_{j}, \: \rightw{\bR}_{j}: ~~ N\times N, \quad j=\overline{2,J}. \qquad \qquad 
\end{aligned} 
\end{equation}
%----------------------------------------------------------------------------------------------------------------------%
\subsection{Factorized wavefield} \label{FacWav}
%----------------------------------------------------------------------------------------------------------------------%

Combining~\eqref{sol1} and~\eqref{sol2}, we obtain the system of algebraic equations 
\begin{equation} \label{sol2bis}
\begin{aligned}
&\text{(a)}& \quad \leftw{\br} &= \leftw{\bT}_{\!1} \leftw{\bbeta}^{\!1} + \rightw{\bR}_1 \rightw{\bi}, \\
&\text{(b)}& \quad \rightw{\bbeta}^{\!1} &= \leftw{\bR}_{1} \leftw{\bbeta}^{\!1} + \rightw{\bT}_{\!1} \rightw{\bi} , \\
&\text{(c)}& \quad  \leftw{\bbeta}^{\nes j-1} &= \leftw{\bT}_{\!\!j} \leftw{\bbeta}^{j} + \rightw{\bR}_{j} \rightw{\bbeta}^{\nes j-1} ,\quad j=\overline{2,J} \\
&\text{(d)}& \quad \rightw{\bbeta}^{\nes j} &= \leftw{\bR}_{j} \leftw{\bbeta}^{\nes j} + \rightw{\bT}_{\!\!j} \rightw{\bbeta}^{\nes j-1}, \quad j=\overline{2,J} \\
&\text{(e)}& \quad \leftw{\bbeta}^{\nes J} &= \rightw{\bR}_{J+1} \rightw{\bbeta}^{\nes J}, \\
&\text{(f)}& \quad \rightw{\bt} &=  \rightw{\bT}_{\!\!J+1} \rightw{\bbeta}^{\nes J},
\end{aligned} 
\end{equation}
where~$\bi$ signifies the source excitation i.e. the incident plane wave, whose angle of incidence is encoded in the featured reflection and transmission matrices. The solution of this system can be conveniently rewritten as 
\begin{equation} \label{sol3}    
\begin{aligned}
\leftw{\br} &= \leftw{\bT}_{\!1} \leftw{\bbeta}^{\!1} + \rightw{\bR}_1 \rightw{\bi}, \\ 
\rightw{\bt} &=  \rightw{\bT}_{\!\!J+1} \rightw{\bbeta}^{\nes J},
\end{aligned} \qquad\text{and}\qquad
\begin{aligned}
\rightw{\bbeta}^{\nes j} &= \rightw{\bmT}_j \, \rightw{\bmT}_{j-1} \ldots \rightw{\bmT}_1 \rightw{\bi}, \\
\leftw{\bbeta}^{\nes j} &= \rightw{\bmR}_{j+1} \rightw{\bbeta}^{\nes j} 
\end{aligned} \quad (j=\overline{1,J}) 
\end{equation}
 where $\rightw{\bmT}_j$ and $\leftw{\bmR}_j$ denote respectively the \emph{generalized} transmission and reflection matrices given by the recursive formulas 
\begin{equation} \label{sol4}    
\begin{aligned}
\rightw{\bmR}_{J+1} &= \rightw{\bR}_{J+1} \\
\rightw{\bmT}_{\!\!j} &= \big(\bmI - \leftw{\bR}_{j}\rightw{\bmR}_{j+1} \big)^{-1} \rightw{\bT}_{\!\!j}, \quad j=J,J\!-\!1,\ldots 1, \\
\rightw{\bmR}_j &= \rightw{\bR}_{j} + \leftw{\bT}_{\!\!j} \rightw{\bmR}_{j+1} \rightw{\bmT}_{\!\!j}, \quad \quad j=J,J\!-\!1,\ldots 1.  
\end{aligned}
\end{equation}
For an algebraic background behind~\eqref{sol4}, the reader is referred to~\citep{guzin2001} dealing with the scattering of plane waves by a layered system composed of homogeneous strips. We note for completeness that the dimension of~$\rightw{\bmT}_{\!\!1}$ is $N\times(2M\nes+\nes 1)$, that of~$\rightw{\bmR}_{1}$ is $(2M\nes+\nes1)\times(2M\nes+\nes 1)$, while the dimension of all other generalized (transmission and reflection) matrices is~$N\times N$. With~\eqref{sol3} in place, vectors~$\leftw{\br}$ and~$\rightw{\bt}$ synthesizing the reflected and transmitted fields are computed respectively via the first and the last of~\eqref{sol2bis}. 

Letting $j=\overline{2,J}$, $\rightw{\bmR}_{j}$ and~$\rightw{\bmT}_{j}$ can be understood as the ``condensed' matrices of reflection and transmission coefficients, respectively, for the \emph{right-going Bloch wave modes} impinging on~$\Gamma_{\!j}$. By design, they include the cumulative effect of multiple scattering generated by the heterogeneous layer between~$\xi_1=x^{(j)}$ and~$\xi_1=L$, expressed in terms of the reflected and transmitted Bloch wave modes emanating from~$\Gamma_{\!j}$. By contrast, $\rightw{\bmR}_1$ describes the cumulative reflection of the \emph{right-going Fourier wave modes} off the first interface~$\Gamma_1$, while $\rightw{\bmT}_{\!\!1}$ is responsible for the cumulative transmission, and thus conversion, of the right-going Fourier wave modes into the Bloch wave modes across~$\Gamma_1$. 

%----------------------------------------------------------------------------------------------------------------------%
\subsection{The vanishing role of strongly evanescent Bloch waves} \label{strongeva}
%----------------------------------------------------------------------------------------------------------------------%

As elucidated in Sec.~\ref{trunceigsp}, we ignore the contribution of strongly decaying Bloch waves by truncating the eigenspectrum of~\eqref{QEP} as to retain only the eigenvalues with ``small'' $|\Im(\kappa_n^j)|$. To justify this premise, it is instructive to recall~\eqref{expomat1} and factorize the featured reflection and transmission matrices as 
\begin{equation} \label{decay1}
\begin{aligned}
\leftw{\bT}_{\!\!j} &= \leftw{\mathds{T}}_{\!\!j} \leftw{\Ereg}^{\!j}_{\!\!j},  \quad 
& \rightw{\bT}_{\!\!j} &= \rightw{\mathds{T}}_{\!\!j} \rightw{\Ereg}^{\!j-1}_{\!\!j} \\
\leftw{\bR}_{j} &= \leftw{\mathds{R}}_{j} \leftw{\Ereg}^{\!j}_{\!\!j}, \quad 
&\rightw{\bR}_{j} &= \rightw{\mathds{R}}_{j} \rightw{\Ereg}^{\!j-1}_{\!\!j}
\end{aligned}, \qquad j=\overline{1,J\!+1} 
\end{equation}
where $\rightw{\Ereg}^{\!0}_{\!1}:= \bI$ and~$\leftw{\Ereg}^{\!J\!+1}_{\!\!J\!+1}:= \bI$. On the basis of~\eqref{solx1}, \eqref{solx3} and~\eqref{decay1}, one finds that~$\leftw{\mathds{T}}_{\!\!j}, \leftw{\mathds{R}}_{j}, \rightw{\mathds{T}}_{\!\!j}$ and~$\rightw{\mathds{R}}_{j}$ are \emph{independent} of the eigenvalues~$\kappa_n^j$. 

\paragraph{Right-going waves.} By way of~\eqref{decay1}, recursive formulas~\eqref{sol4} can be rewritten as 
\begin{equation} \label{sol5}    
\begin{aligned}
&\text{(a)}& \quad \rightw{\bmR}_{J+1} &= \rightw{\mathds{R}}_{\nes J\!+1} \rightw{\Ereg}^{\!J}_{\!\!J\!+1} \\
&\text{(b)}& \quad \rightw{\bmT}_{\!\!j} &= \big(\bmI - \leftw{\mathds{R}}_{j} \leftw{\Ereg}^{\!j}_{\!\!j}\rightw{\bmR}_{j+1} \big)^{-1} \rightw{\mathds{T}}_{\!\!j} \rightw{\Ereg}^{\!j-1}_{\!\!j}, \quad j=J,J\!-\!1,\ldots 1, \\
&\text{(c)}& \quad \rightw{\bmR}_j &= \rightw{\mathds{R}}_{j} \rightw{\Ereg}^{\!j-1}_{\!\!j} +  \leftw{\mathds{T}}_{\!\!j} \leftw{\Ereg}^{\!j}_{\!\!j} \rightw{\bmR}_{j+1} \rightw{\bmT}_{\!\!j}, \quad \quad j=J,J\!-\!1,\ldots 1.  
\end{aligned}
\end{equation}
From~(\ref{sol2bis}a), (\ref{sol2bis}f), (\ref{sol3}) and~(\ref{sol5}b), it is clear that 
\begin{equation} \label{sol5x}    
\begin{aligned}
\rightw{\bbeta}^{\nes 1} &= \rightw{\bmT}_{\!\!1}\! \rightw{\bi} \,=\, (\ldots) \rightw{\mathds{T}}_{\!\!1} \!\rightw{\bi}, \qquad &\rightw{\bbeta}^{\nes j} &= \rightw{\bmT}_{\!\!j} \rightw{\bbeta}^{\nes j-1} = (\ldots) \rightw{\Ereg}^{\!j-1}_{\!\!j}\rightw{\bbeta}^{\nes j-1},   \\
\leftw{\bbeta}^{\nes 1} &= \rightw{\bmR}_2 \rightw{\bbeta}^{\nes 1} = (\ldots) \rightw{\Ereg}^{\!1}_{\!\!2}  \rightw{\bbeta}^{\nes 1}, \qquad &\leftw{\bbeta}^{\nes j} &= \rightw{\bmR}_{j+1}\rightw{\bbeta}^{\nes j}  = 
(\ldots) \rightw{\Ereg}^{\!j}_{\!\!j+1}\rightw{\bbeta}^{\nes j},
\end{aligned} \quad j=\overline{2,J} \qquad 
\end{equation}
and 
\begin{equation} \label{IR1}
\begin{aligned}
\leftw{\br} &= \leftw{\bT}_{\!1} \leftw{\bbeta}^{\!1} + \rightw{\bR}_1 \! \rightw{\bi} \,=\, 
(\ldots) \rightw{\Ereg}^{\!1}_{\!\!2}  \rightw{\bbeta}^{\nes 1} + \rightw{\mathds{R}}_{1} \!\rightw{\bi}, \\
\rightw{\bt} &=  \rightw{\bT}_{\!\!J+1} \rightw{\bbeta}^{\nes J} \,=\, 
\rightw{\mathds{T}}_{\!\!J+1} \rightw{\Ereg}^{\!J}_{\!\!J+1} \rightw{\bbeta}^{\nes J},
\end{aligned}
\end{equation}
where $\rightw{\Ereg}^{\!j-1}_{\!\!j}$ ($j=\overline{2,J\!+1}$) is given by~\eqref{expomat1}. This demonstrates that the right-going Bloch eigenstates affiliated with large positive values of~$\Im(\rightw{\kappa}^{j-1}_{\!n})$, which generate vanishingly small diagonal entries $e^{i \rightw{\kappa}^{j-1}_{\!n}{w}^{(j-1)}}$ of $\rightw{\Ereg}^{\!j-1}_{\!\!j}$, can be omitted in the truncated set~$\rightw{\bbeta}^{\nes j-1}$ ($j=\overline{2,J\!+1}$) for their contribution to the right-hand sides in \eqref{sol5x}--\eqref{IR1} is negligible. 

\paragraph{Left-going waves.} Let $\leftw{\bbeta}^{\nes j}$ ($j=\overline{1,J}$) be computed from~\eqref{sol5}--\eqref{sol5x}. To expose the role of strongly attenuated left-going Bloch eigenstates contained in~$\leftw{\bbeta}^{\nes j}$, the above solution can be conveniently re-factorized~\citep{guzin2001} as 
\begin{equation} \label{sol5y}    
\rightw{\bbeta}^{\nes j} = \leftw{\bmR}_{\nes j} \leftw{\bbeta}^{\nes j} \,=\, (\ldots) \leftw{\Ereg}^{\!j}_{\!\!j} \leftw{\bbeta}^{\nes j},  \qquad j=\overline{1,J}  
\end{equation}
and
\begin{equation} \label{IR2}
\begin{aligned}
\leftw{\br} &= \leftw{\mathds{T}}_{\!\!1} \leftw{\Ereg}^{\!1}_{\!\!1} \leftw{\bbeta}^{\!1} + \rightw{\mathds{R}}_{1} \!\rightw{\bi}, \\
\rightw{\bt} &=  \rightw{\bT}_{\!\!J+1} \leftw{\bmR}_{\nes J} \leftw{\bbeta}^{\nes J} \,=\, (\ldots) \leftw{\Ereg}^{\!J}_{\!\!J} \leftw{\bbeta}^{\nes J},
\end{aligned}
\end{equation}
where 
\begin{equation} \label{sol7}    
\begin{aligned}
\leftw{\bmR}_{1} &= \leftw{\mathds{R}}_{1} \leftw{\Ereg}^{\!1}_{\!\!1}, \\
\leftw{\bmT}_{\!\!j} &= \big(\bmI - \rightw{\mathds{R}}_{j} \rightw{\Ereg}^{\!j-1}_{\!\!j} \leftw{\bmR}_{j-1} \big)^{-1} \leftw{\mathds{T}}_{\!\!j} \leftw{\Ereg}^{\!j}_{\!\!j}, \quad j=2,3,\ldots J\!+1, \\
\leftw{\bmR}_j &= \leftw{\mathds{R}}_{j} \leftw{\Ereg}^{\!j}_{\!\!j} + 
\rightw{\mathds{T}}_{\!\!j} \rightw{\Ereg}^{\!j-1}_{\!\!j} \leftw{\bmR}_{j-1} \leftw{\bmT}_{\!\!j}, \quad \quad j=2,3,\ldots J\!+1
\end{aligned}
\end{equation}
and $\leftw{\Ereg}^{\!j}_{\!\!j}$ ($j=\overline{1,J}$) is given by~\eqref{expomat1}. Analogous to the previous argument, we observe from \eqref{sol5y}--\eqref{IR2} that the left-going Bloch eigenstates having large negative values of~$\Im(\leftw{\kappa}^{j}_{\!n})$, which generate vanishingly small diagonal entries $e^{-i \leftw{\kappa}^{j}_{\!n}{w}^{(j)}}$ of $\rightw{\Ereg}^{\!j}_{\!\!j}$, can be omitted in the truncated set~$\leftw{\bbeta}^{\nes j}$ ($j=\overline{1,J}$) without having an appreciable effect on the solution. 

%----------------------------------------------------------------------------------------------------------------------%
\subsection{Plug-and-play simulation paradigm} \label{plug-and-play}
%----------------------------------------------------------------------------------------------------------------------%

Let~$\Omega$ denote a layer-cake system comprising~$J$ strips of $Q\!\leqslant\!J$ periodic media $\mathcal{M}_q$ ($q\shh \overline{1,Q}$) that is being designed for a particular facet of wave manipulation, e.g. rainbow trapping or wave steering. From the foregoing developments, one finds that for given $\omega$ and~$k_2$, $Q$ solutions of the quadratic eigenvalue problem~\eqref{QEP} -- computed respectively for~$\mathcal{M}_q$ ($q\shh \overline{1,Q}$) -- provide a basis for rapid exploration of the design space of $\Omega$ that includes:
\begin{itemize}[leftmargin=*]
\item The number of strips, $J$;
\item Strip makeups~$\mathfrak{m}_j\!\in\overline{1,Q}$, where $\mathfrak{m}_j\shh q$ designates  the $j$th strip ($j\shh \overline{1,J}$) as a cutout from~$\mathcal{M}_q$;
\item Strip widths $w^{(j)}=x^{(j+1)}-x^{(j)}$, $\, j\shh \overline{1,J}$;
\item Translations~$\bs_{\nes j}$ of the periodic medium~$\mathcal{M}_{\mathfrak{m}_j}$ relative to the unit cell window~$Y_{\!j}$, see Sec.~\ref{sectran}.
\end{itemize}
For instance, variations in the set $\{\mathfrak{m}_j\}$ can be designed to represent permutations (with or without repetitions) of the strips within $\Omega$, while the application of~$\bs_{\nes j}$ inherently affects the ($\xi_2$-dependent) impedance contrast between the laminae and so provides a major handle on the reflected and transmitted fields. 

In principle, each such alteration entails \emph{anew} (i) computation of the 1D Fourier series in~\eqref{fields2}; (ii) evaluation of the reflection and transmission matrices~\eqref{solx1}, and (iii) computation of the \emph{generalized} reflection and transmission matrices~\eqref{sol4}. In particular, the computational effort underpinning~(ii) and~(iii) entails low-dimensional matrix manipulations with the highest matrix dimension being $4M\!+\nes 2\geqslant 2N$, where $2N$ is the number of Bloch modes and $2M\!+\nes1$ is the number of Fourier modes.

As examined earlier, the linchpin of the proposed approach are the Bloch eigenvalues~$\kappa_n^q$ and eigenfunctions~$\phi_n^q$ of the QEVP~\eqref{QEP} (computed for given~$\omega, k_2$ and $M_q, \, q\shh \overline{1,Q}$) that are \emph{common} across the featured design space. Due to the fact that $\kappa_n^q\!\in\mathbb{C}$, the Bloch-wave expansion of the fields within each strip provides a natural basis for capturing the \emph{boundary layers} stemming from the interfaces~$\Gamma_{\!j}$, $j\shh \overline{1,J\!+\!1}$ in the form of (pure or leaky) evanescent waves. 

For sufficiently high number~$2N$ of the Bloch wave modes (which then drives the number $2M\!+\nes1$ of Fourier modes), the factorized Bloch approach is capable providing a high-fidelity solution of the scattering problem. Thus, the accuracy of the proposed scheme revolves around the choice of~$N$, that is for simplicity taken as common for all strips. On the basis of~\eqref{sol5x}--\eqref{IR2}, once the quadratic eigenvalues~$\kappa_n^q$ are computed for all $M_q$ $(q\shh \overline{1,Q}$), this parameter can be selected \emph{a priori} across the featured design space. Specifically, if one specifies $w_{\text{min}}\!>\!0$ as the lower bound on the allowable variations on~$w^{(j)}$, $N$ could be selected from the requirement that 
\begin{equation}\label{Ncrit}
|e^{i \rightw{\kappa}^{q}_{\!\nes 2N} w_{\text{min}}}| \,\leqslant\, \epsilon \, |e^{i \rightw{\kappa}^{q}_{\!\! N+1} w_{\text{min}}}| \quad \text{and} \quad
|e^{-i \leftw{\kappa}^{q}_{\!\!N} w_{\text{min}}}|
\,\leqslant\,  \epsilon \, |e^{-i \leftw{\kappa}^{q}_{\!\nes 1} w_{\text{min}}}|, \qquad q=\overline{1,Q}
\end{equation}
for some prescribed $\epsilon\shh o(1)$. The motivation behind~\eqref{Ncrit} is simple: we consider a $w_{\text{min}}$-wide strip of~$\mathcal{M}_q$, then $|e^{i \rightw{\kappa}^{q}_{\!\!N+1} w_{\text{min}}}|\leqslant 1$ signifies the reduction in amplitude of the \emph{least-decaying}, right-going wave due to propagation across the strip. Then, the first of~\eqref{Ncrit} guarantees that every discarded Bloch mode will have the amplitude reduction bounded from above by $\epsilon\,|e^{i \rightw{\kappa}^{q}_{\!\!N+1} w_{\text{min}}}|$. Analogous argument applies to the second of~\eqref{Ncrit}. Of course, other criteria are also applicable, for instance by requiring that $|e^{i \rightw{\kappa}^{q}_{\!\nes 2N} w_{\text{min}}}| \leqslant \epsilon$ and $|e^{-i \leftw{\kappa}^{q}_{\!\!N} w_{\text{min}}}|\leqslant \epsilon$, i.e. by introducing a common upper bound on the amplitude reduction of discarded Bloch modes for all strips. 

%----------------------------------------------------------------------------------------------------------------------%
\section{Numerical results}\label{numres}
%----------------------------------------------------------------------------------------------------------------------%

In what follows, numerical simulations are performed via NGSolve -- an open-source, Python-based finite element (FE) computational platform~\citep{NGSolve}. For generality, all physical parameters are presented in a dimensionless form – as normalized by  the characteristic size of the unit cell, the mass density of the background medium, and the shear modulus of the background medium. With reference to Sec.~\ref{plug-and-play}, we illustrate exploration of the design space that entails (i) permutation (without repetitions) of the strip makeups i.e. the set~$\{\mathfrak{m}_j\}$, and (ii) translations~$\bs_{\nes j}$ of the periodic media~$\mathcal{M}_{\mathfrak{m}_j}$ ($j\shh \overline{1,J}$), each featuring a square unit cell ($\ell^{(j)}\!= d= 1$). 

%----------------------------------------------------------------------------------------------------------------------%
\subsection{Quadratic eigenspectrum}\label{numQEVP}
%----------------------------------------------------------------------------------------------------------------------%

To illustrate the analysis, we first compute a truncated spectrum $\{\kappa_n\}$ of the the QEVP~(\ref{QEP}) for the example unit cell shown in Fig.~\ref{QEVP_FE}(a) by letting $(\omega, k_2)\shh  (2, 1)$ and meshing the unit cell with order-five ($p\shh 5$) triangular elements whose maximum characteristic size is $h\shh 0.05$, i.e.~$d/20$. The eigenvalues are obtained by first converting~(\ref{QEP}) to a larger-size, linear eigenvalue problem~\citep{Lack2019} and then solving the latter via Arnoldi iterations~\citep{Leho1998}. The first 200 ``unfiltered'' eigenvalues $\kappa_n\!\in\mathbb{C}$ are plotted in Fig.~\ref{QEVP_FE}(b), while their restriction to the essential strip~\eqref{essenstrip} is shown in Fig.~\ref{QEVP_FE}(c). From the displays, one observes that (i) the eigenvalues -- which spread radially with increasing~$n$ -- indeed appear as complex-conjugate pairs as demonstrated in Proposition~\ref{prop1}, and (ii) in the essential strip, there are only two real eigenvalues ($\kappa_n\!\in\mathbb{R}$) for the example under consideration; see also the comment at the end of Sec.~\ref{essential_strip}. 

In terms of the computational effort required to evaluate the first $\mathfrak{N}$ ``unfiltered'' eigenvalues, whose restriction to the essential strip yields~$2N$ eigenvalues $\rightw{\kappa}_n$ ($n\shh \overline{1,N}$) and $\leftw{\kappa}_n$ ($n\shh \overline{N\!+\! 1,2N}$), it is found from numerical experiments that the number of Arnoldi iterations, $\eta$, which ensures accurate computation of the sought set is $\eta\simeq 2\mathfrak{N}$. Numerical simulations also exposed a rule of thumb by which taking~$\mathfrak{N} = 20 N$ furnishes at least~$2N$ eigenvalues in the essential strip. 

\begin{figure}
\centering{\includegraphics [width=0.96\textwidth ]{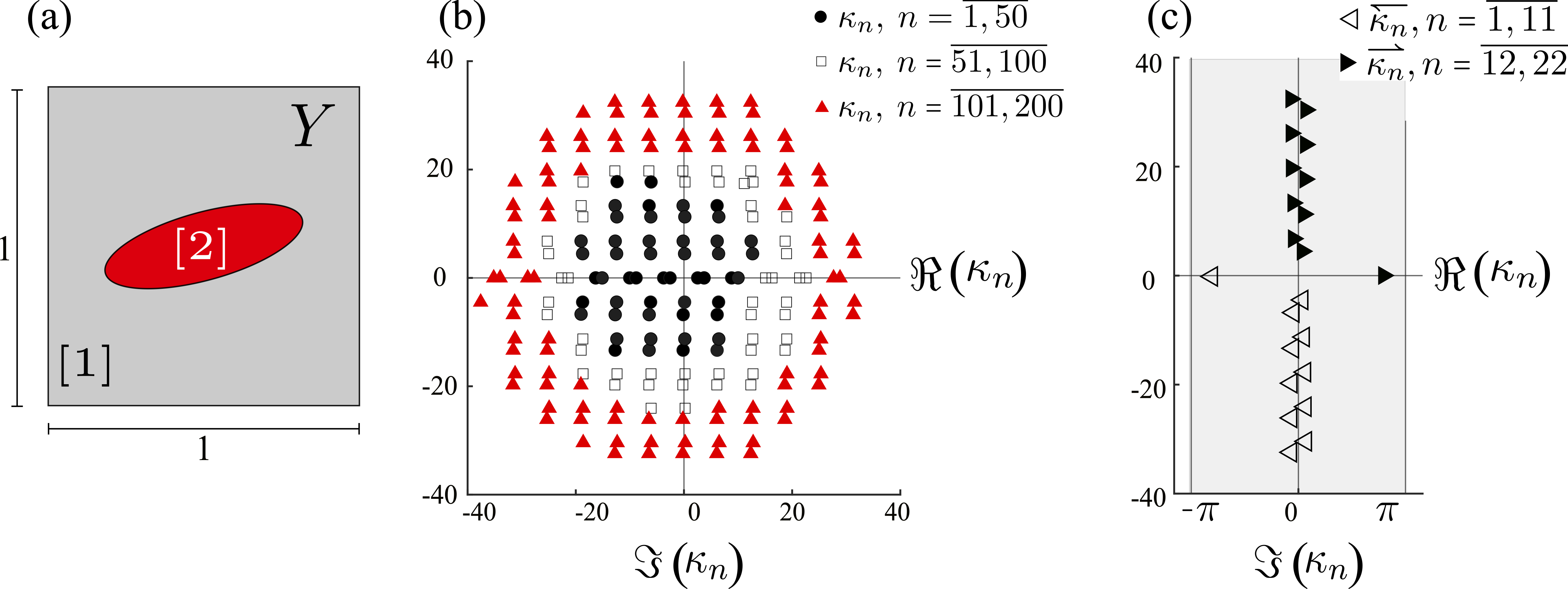}}
\caption{(a) Unit cell of periodicity featuring an elliptical inclusion with $G\jsup{1}=1$, $\rho\jsup{1}=1$, $G\jsup{2}=5$ and $\rho\jsup{2}=0.2$; (b) eigenvalues $\kappa_n$ of the QEVP~(\ref{QEP}) computed for $(\omega, k_2) = (2, 1)$; and (c) restriction of~(b) to the essential strip~\eqref{essenstrip} with $\kappa_n$ designated as either \emph{right-going} (filled triangles) or \emph{left-going} (hollow triangles).} \label{QEVP_FE}
\end{figure}

%----------------------------------------------------------------------------------------------------------------------%
\subsection{Verification of the factorized Bloch solution}\label{numRT}
%----------------------------------------------------------------------------------------------------------------------%

To compare the wavefield in a layer-cake system evaluated via the factorized Bloch solution against full-field FE simulations, we consider a 2D rainbow trap comprising six unit-cell-wide strips ($J=6$) extracted from mother periodic media $\mathcal{M}_j$ ($j=\overline{1,6}$) as shown in Fig.~\ref{RT_dispersion}(a). Following the usual approach in the design of rainbow traps~\citep{Zhu2013,Tian2017}, each $\mathcal{M}_j$ is designed such that its bandgap is shifted slightly ``upward'' relative to that featured by $\mathcal{M}_{j-1}$, see Fig.~\ref{RT_dispersion}(b). For completeness, columns 2 through 7 in Table~\ref{mat_properties} summarize the material properties of the unit cells $Y_j$ ($j=\overline{1,6}$).
\begin{figure}
\centering{\includegraphics [width=0.75\textwidth ]{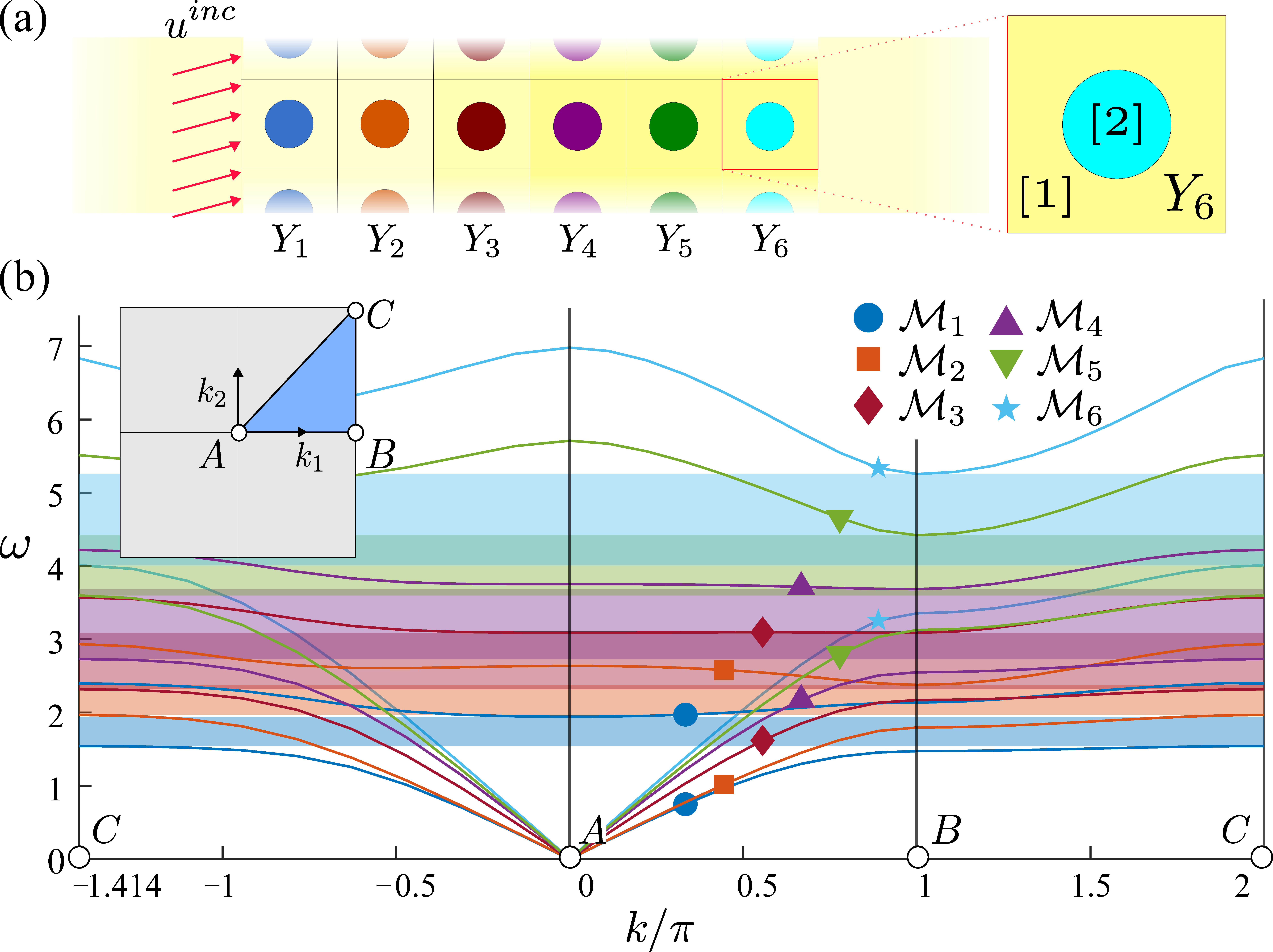}}
\caption{(a) Rainbow trap (1-periodic in the $\xi_2$-direction) with six layers sandwiched between two homogeneous half spaces, and (b) the first two dispersion branches for each of the mother periodic media $\mathcal{M}_j$  ($j=\overline{1,6}$) with the respective bandgaps indicated by shaded strips. Note that $\omega= 1$ resides in the passband of all mother periodic media, while $\omega= 2$, $\omega= 3$ and $\omega= 4$ are contained respectively by the bandgaps of the periodic media $j\nes\in\nes \{1,2\}$,  $j\nes\in\nes \{ 3,4\}$, and $j\nes\in\nes \{5,6\}$.}\label{RT_dispersion}
\end{figure}

\begin{table}[hb]
\caption{Piecewise-constant material properties of the periodic media featured in numerical simulations: rainbow trap example, Sec.~\ref{numRT} (columns 2--7), and optimal design of the layer-cake system, Sec.~\ref{numopt} (columns 8--10). Here, superscript [1] (resp.~[2]) refers to the material properties of the inclusion (resp. background).} 
\centering
\begin{tabular}{|c|c|c|c|c|c|c||c|c|c|} \hline
{Parameter}  & $j=1$ & $2$ & $3$ & $4$ & $5$ & $6$ & $j=1$ & $2$ & $3$ \\ \hline\hline
\rule{0pt}{10pt} {$\rho_j\jsup{1}$}    &  1  &  1  &  1  &  1  &  1  &  1 & 0.5 & 2 & 0.8 \\ \hline
\rule{0pt}{10pt} {$G_j\jsup{1}$} &  1  &  1  &  2  &  3  &  3  &  3 &  2  & 1 &  3  \\ \hline
\rule{0pt}{10pt} {$\rho_j\jsup{2}$}    &  5  &  5  &  6  &  7  &  7  &  7 &  5  & 5 &  5  \\ \hline 
\rule{0pt}{10pt} {$G_j\jsup{2}$} & 0.1 & 0.2 & 0.3 & 0.5 & 1.5 &  4 & 0.1 & 2 & 0.2 \\ \hline 
\end{tabular}
\label{mat_properties}
\end{table}

Assuming an oblique incident plane wave impinging ``from below'' on the rainbow trap at an angle $\theta=\pi/6$ relative to the $\xi_1$-axis for which $\bk= \omega(\sqrt{3}/2,1/2)$, time-harmonic FE simulations are performed at four frequency-wavenumber combinations
\[
(\omega, k_2)  = (\omega,\omega/2), \quad \omega = \overline{1,4}. 
\] 
To reduce the computational effort, the FE simulations are designed to solve for the \emph{factored} displacement field $\tilde{u}(\bxi) = u(\bxi) \, e^{-i k_2 \xi_2}$ as in~\eqref{factusig}, which satisfies the field equation
\[
\nabla_{\!\bk'}  \sip (G(\bxi) \nabla_{\!\bk'} \tilde{u}) + \omega^2 \rho(\bxi) \hh \tilde{u} \,=\, 0, \qquad \bk'=(0, k_2).
\]
The advantage of such an approach is that $G,\rho,\tilde{u}^{inc} \!= e^{i k_1\xi_1}$, and~$\tilde{u}$ are all $d$-periodic in the $\xi_2$-direction, whereby the simulation domain can be reduced to a single strip of ``height'' $d\shh 1$, subject to the periodic boundary conditions along its top and bottom boundary as shown in Fig.~\ref{RT_validation}(a). The left and the right half-space are each modeled by a homogeneous strip ($G_\circ\shh 1,\rho_\circ\shh 1$) of width $10$, followed by an absorbing boundary implemented via perfectly matched layers~\citep{NGSolve}. 

With reference to Fig.~\ref{RT_validation}(a), the factorized incident wavefield $\tilde{u}^{inc}$ is generated by a set of~$P$ equidistant point sources placed along a vertical line in the middle of the left homogeneous strip. The Dirac delta function affiliated with each point source is approximated by Gaussian distribution with variance $\gamma^2$. The parameter $P$ (resp.~$\gamma$) is chosen such that $(P+1) \lambda$ (resp. $\gamma/\lambda$) remains constant at all frequencies, where $\lambda\shh 2\pi/\omega$ denotes the wavelength in each half-space. For the present example, we select $(P+1)\lambda=200$ and $\gamma/\lambda=0.0025$ which ensures veracious approximation of a normally-incident plane wave. The simulation domain is meshed with triangular elements ($p=3$) of maximum characteristic size $h\shh d/20$ which is sufficient to accurately simulate the factorized wave motion $\tilde{u}$ for~$\omega\in[1,4]$. The four panels in Fig.~\ref{RT_validation}(b) plot respectively the total motion $u(\bxi) = \tilde{u}(\bxi) \, e^{i k_2 \xi_2}$ at $\omega=\overline{1,4}$. As can be seen from the display, the factorized Bloch-wave (FB) solution is practically indistinguishable from the FE simulations in all cases. For excitation frequencies $\omega=\overline{2,4}$ that are contained within the union of the band gaps $\omega\in(1.5,5.5)$ of the mother periodic media (see Fig.~\ref{RT_dispersion}(b)), one also observes the effect of rainbow trapping where the transmitted waves in the right half-space are largely diminished. 

\begin{figure}[ht!]
\centering{\includegraphics [width=0.8\textwidth]{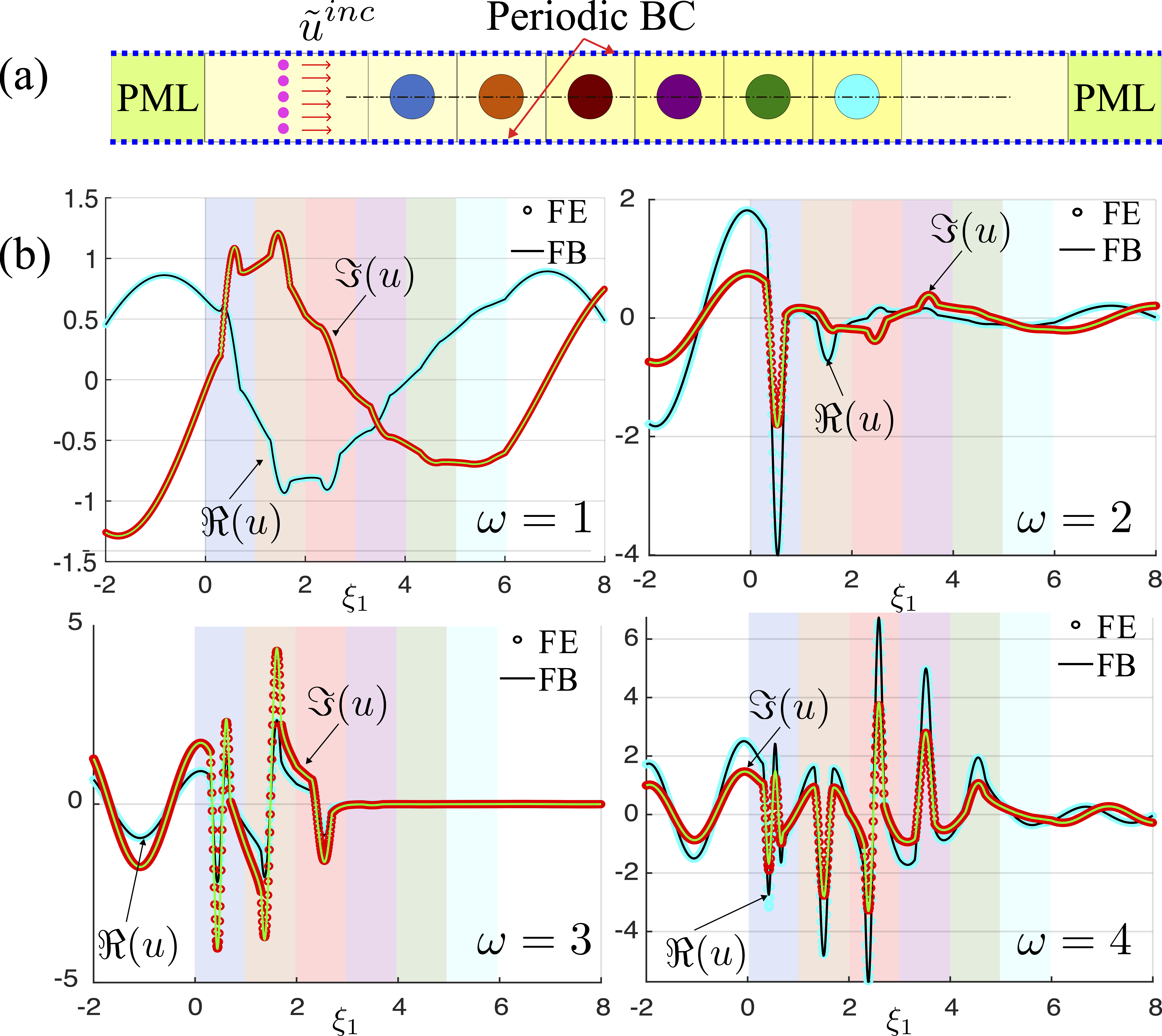}}
\caption{(a) Numerical model of the rainbow trap ($0\!<\!\xi_1\!<\!6$) and adjacent half-spaces used for full-scale FE simulations. The (top and bottom) dashed lines indicate periodic boundary conditions, and the dots in the left homogeneous medium depict point sources used to generate the normally-incident (factorized) plane wave. The simulation domain is sandwiched between two absorbing boundaries implemented as perfectly matched layers (PMLs). (b) Comparison of the factorized Bloch (FB) solution versus FE simulations along the centerline of the simulation domain, computed for $\omega= \overline{1,4}$.} \label{RT_validation}
\end{figure}
\nopagebreak
For completeness, we recall that computing the factorized Bloch-wave solution at each~$\omega$ entails calculating a truncated quadratic eigenspectrum $\{\kappa^q_n, \phi^q_n\}_{n=1}^{2N}$ (see Sec.~\ref{numQEVP}) for each mother periodic medium~$\mathcal{M}_q$ $(q=\overline{1,6})$, and then evaluating the factorized wavefield as in Sec.~\ref{FacWav}.  With reference to the discussion in Sec.~\ref{plug-and-play}, in simulations we take $N=9$ which, for the problem under consideration with $w_{\text{min}}=1$, ensures that all ``weakly-evanescent'' Bloch modes captured by the respective conditions
\[
|e^{i \rightw{\kappa}^{q}_{\!\nes 2N}}| \leqslant \epsilon \quad \text{and} \quad  |e^{-i \leftw{\kappa}^{q}_{\!\!N}}| \leqslant \epsilon, \qquad \epsilon=10^{-2} \quad (q=\overline{1,6})
\]
are included in the expansion. Guided by the requirement $4M\!+\nes 2\geqslant 2N$ elucidated in Remark~\ref{MN-relat}, the factorized Bloch solution in Fig.~\ref{RT_validation} is computed with $M\!=4$. From numerical experiments, it is found that having either even-determined or slightly-overdetermined system in that $4M\!+\nes 2\nes-\nes 2N\ll 4M\!+\nes 2$ provides for an optimal choice of~$M$. Specifically, taking $4M\!+\nes 2$ to be notably larger than $2N$ impedes the ability of the available $2N$ Bloch waves to match the low-Fourier-order (displacement and traction) continuity conditions between strips with sufficient accuracy. 

%----------------------------------------------------------------------------------------------------------------------%
\subsection{Optimal design of a layer-cake system} \label{numopt}
%----------------------------------------------------------------------------------------------------------------------%

To illustrate the utility of the plug-and-play simulation paradigm toward efficient exploration of the design space of layer-cake periodic systems, we consider a system consisting of three strips extracted respectively from the mother periodic media~$\mathcal{M}_q$ ($q\shh \overline{1,3}$) and pursue an ``optimal design'' of the so-constrained system that, for a given incident wave, minimizes the amplitude of body waves transmitted into the right half-space. Recalling~\eqref{usigmp1}, we specifically seek to minimize the amplitude of the propagative (i.e.~non-evanescent) component of~$\tilde{u}^+$ and thus~$u^+$ given by 
\begin{equation} \label{Tcal}
\mathcal{T} :=\,  \Big(\sum_{m=-M}^{M} c_m\, |t_m|^2\Big)^{1/2}, \qquad 
c_m \,=\, \left\{ \begin{array}{ll}
1, &  k_{1m}^+\in \mathbb{R} \\ 0, & k_{1m}^+\notin \mathbb{R} \end{array} \right. .
\end{equation}
For completeness, we note that the amplitude $\mathcal{R}$ of the propagative reflected field is obtained by replacing~$t_m$ and~$k_{1m}^+$ in~\eqref{Tcal} respectively by~$r_m$ and~$k_{1m}^-$, see~\eqref{usigmp1}. 

We start with the three-strip system shown in Fig.~\ref{optimal_surfaces}(a), whose material properties are given in the last three columns of Table~\ref{mat_properties} (recall that $G_\circ\shh 1$ and $\rho_\circ\shh 1$). Letting $J\shh Q\shh 3$, we restrict the design space to (i) permutation of the strip makeups $\mathfrak{m}_j\in\overline{1,Q}$ ($j\shh\overline{1,J}$) without repetitions, and (ii) translations~$\bs_{\nes j}$ of the periodic medium~$\mathcal{M}_{\mathfrak{m}_j}$ relative to the unit cell window~$Y_{\!j}$. Such design constraint amounts to (among other things) keeping the volume fractions of all material components in the  system constant. 

We first consider permutations of the mother periodic media, and identify the optimal configuration $C_1$ (obtained by permutations without translations) shown in Fig.~\ref{optimal_surfaces}(b) that reduces the initial value of~$\mathcal{T}$ by 14.6\%. To illustrate the effect of engaging~$\bs_j$, for each permutation we compute the variation of~$\mathcal{T}$ under translations of the mother periodic medium in a \emph{single strip} of the system; for instance we let~$\bs_1$ vary over $[0,1)\times[0,1)$ while keeping $\bs_2\shh\bs_2\shh\bzero$, and so on. The resulting diagrams are plotted in Fig.~\ref{optimal_surfaces}(d), where the three rows of surfaces (top to bottom) depict the variations of~$\mathcal{T}$ vs.~$\bs_1$ through~$\bs_3$ obtained by covering the domain $[0, 1) \times [0, 1)$ with $100 \times 100$ uniformly-spaced grid points. By selecting the layer-cake configuration that produces the minimum $\mathcal{T}$-value between these 18 diagrams, we identify the optimal configuration~$C_2$ given by (i) the permutation shown in Fig.~\ref{optimal_surfaces}(b), and (ii) translation $\bs_2\shh(0.39,0.25)$ of the mother periodic medium in the second strip, which reduces the initial $\mathcal{T}$-value by 99.4\%. 

We next focus on the permutation featured by~$C_2$ ($\mathfrak{m}_1\shh 2, \mathfrak{m}_2\shh 1, \mathfrak{m}_3\shh 3$), and compute the variation of the transmission coefficient~$\mathcal{T}$ under \emph{simultaneous} translations $\bs_j \nes\in\nes  [0,1) \times [0,1), \; j=\overline{1, 3}$ of the  mother periodic media for all three strips. In this case, we cover the translation domain $[0,1) \times [0,1)$ for each strip by a sparse grid of $7 \times 7$ uniformly-spaced translations, which produces 117649 combinations. The optimal configuration~$C_3$ identified in this way, shown in Fig.~\ref{optimal_surfaces}(c), performs similar to~$C_2$ and reduces the initial~$\mathcal{T}$-value by 99.4\%; it features the translations $\bs_1\shh\bzero, \bs_2\shh(0.39,0.25)$ and~$\bs_3\shh(0.79,0.60)$. For completeness, Table~\ref{optim_coeff} lists the values of the transmission coefficient~$\mathcal{T}$ and reflection coefficient~$\mathcal{R}$ for each optimal configuration. Here it is worth noting that $\mathcal{T}^2\!+\mathcal{R}^2 = 1+\delta$ with $|\delta|<\num{4e-4}$ for configurations $C_{0/1}$, and \mbox{$|\delta|<\num{1.3e-2}$} for configurations $C_{2/3}$. In the latter case, diminished accuracy of the solution stems from the translation-induced discontinuities in material properties along either side of the strip interfaces (see Fig.~\ref{optimal_surfaces}(c) for example), which in principle require higher~$M$ (and thus higher~$N$) for high-fidelity enforcement of the interfacial conditions. The reason behind the identity $\mathcal{T}^2\!+\mathcal{R}^2 = 1$ is exposed by considering the power flow of (horizontally-propagating) factorized fields $\tilde{u}^{inc}$, $\tilde{u}^r=\tilde{u}^--\tilde{u}^{inc}$ and $\tilde{u}^t=\tilde{u}^+$ within the strip in Fig.~\ref{RT_validation}(a) through sections $\xi_1\shh \pm \infty$, where the amplitudes of the three fields respectively equal 1, $\mathcal{R}$ and~$\mathcal{T}$ owing to a vanishing contribution of the evanescent waves there.

\begin{figure}[ht!]
\centering{\includegraphics[width=1.0\textwidth]{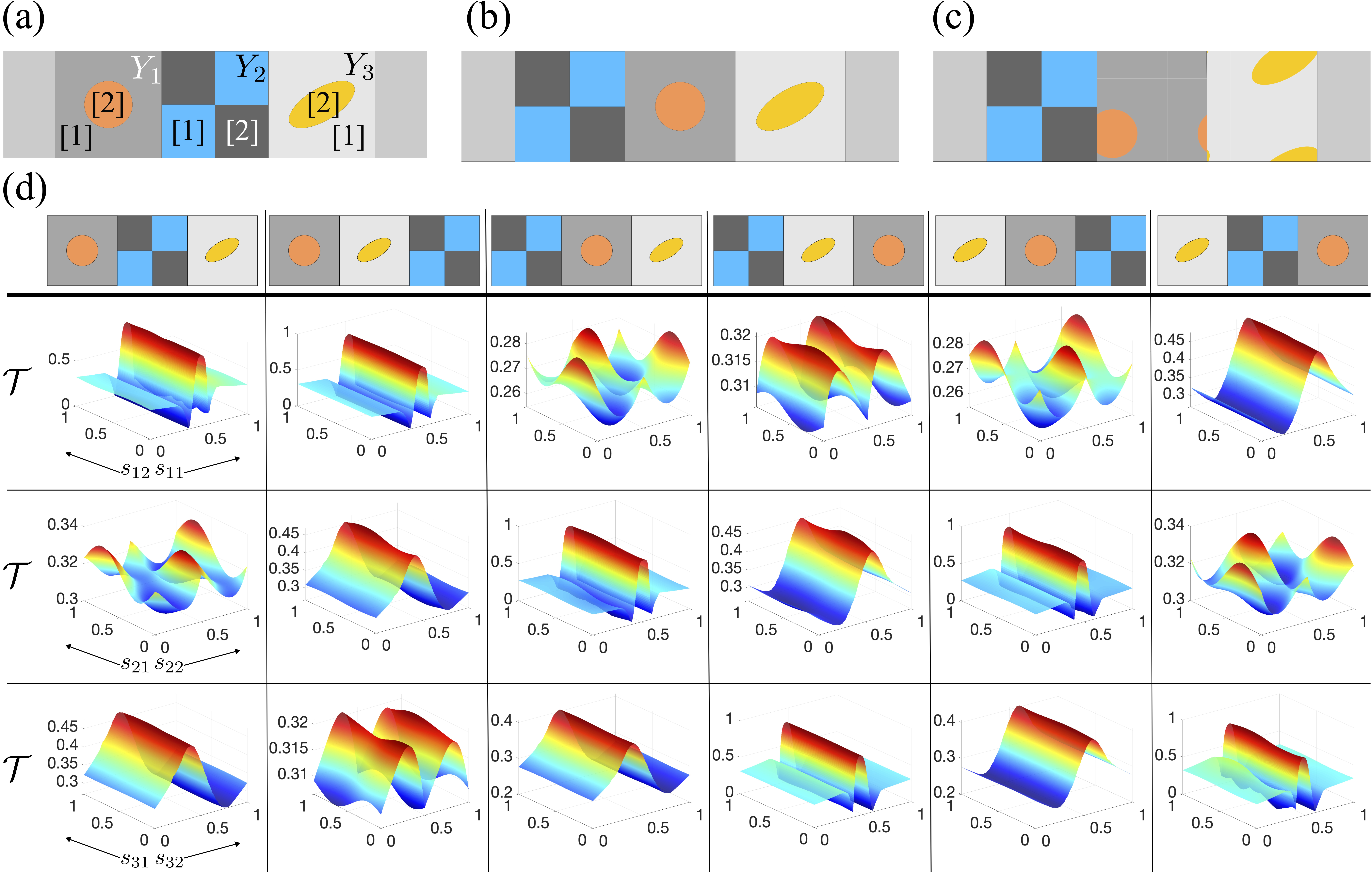}}
\caption{(a) Initial layer-cake system, intended to suppress wave transmission; (b) optimal permutation without translations; (c) optimal configuration including translations~$\bs_j$ ($j\!=\!\overline{1,3}$); and (d) transmission surfaces ($\mathcal{T}$ vs.~$\bs_j$) generated for each permutation by varying exclusively $\bs_1$ (top row), $\bs_2$ (middle row), and $\bs_3$ (bottom row).} \label{optimal_surfaces}
\end{figure}

\begin{table}[hb]
\caption{Transmission ($\mathcal{T}$) and reflection ($\mathcal{R}$) coefficients for the layer-cake configurations in  Figs.~\ref{optimal_surfaces}(a)--(c) generated by permutations and window translations of the mother periodic media.} 
\centering
\begin{tabular}{|l|c|c|c|} \hline
Layer-cake configuration & Fig. & $\mathcal{T}$ & $\mathcal{R}$ \\ \hline\hline 
Initial $C_0$ & \ref{optimal_surfaces}(a) &   0.3231 & 0.9464 \\ \hline
Optimal $C_1$: permutations only & \ref{optimal_surfaces}(b) & 0.2758 & 0.9614 \\ \hline
Optimal $C_2$: permutations w/\emph{individual} translations $\bs_j$ ($j\nes\in\nes \{1,2,3\})$ & ----- &  0.0019 & 1.0063 \\ \hline
Optimal $C_3$: permutation~$C_2$ w/\emph{simultaneous} translations $\bs_j$ ($j\shh\overline{1,3}$) &  \ref{optimal_surfaces}(c) & 0.0018 & 1.0063  \\ \hline 
\end{tabular}
\label{optim_coeff}
\end{table}

\subsection{Discussion}  With the use of suitable search techniques, e.g. genetic algorithms, the design space examined above can be readily expanded as examined in Sec.~\ref{plug-and-play} to include (a) permutations with repetitions, and (b) alterations of the strip widths $w^{(j)}$, $\hh j\shh \overline{1,J}$. In such situations, the minimization could be endowed with a penalty term entailing e.g. the overall thickness, $\sum_{j=1}^{J} w^{(j)}$, of the layer-cake system. 

Concerning the computational savings brought about by the factorized Bloch (FB) paradigm, we note that a single FE simulation of wave scattering by the rainbow trap in Fig.~\ref{RT_validation}(a) for $\,\omega\in[1,4]$ takes $7$ seconds on a reference laptop computer. By contrast, the upfront effort of solving the QEVP~\eqref{QEP} for each of the mother periodic media $\mathcal{M}_j$ ($j\shh \overline{1,6})$ takes~$3$ minutes, while FB evaluation of the scattering problem for a single trial configuration in Sec.~\ref{numopt} takes~$0.3$ seconds. As a result, the FB approach provides $23\times$ computational speedup in performing 117649 simulations  behind the optimal configuration~$C_3$ shown in Fig.~\ref{optimal_surfaces}(c). Inherently, this speedup factor increases by decades when considering (i) an expanded design space (by allowing for variable strip widths and permutations with repetitions), and (ii) higher excitation frequencies as they disproportionately increase the cost of FE simulations. 

Beyond rainbow trapping, optimal design of layer-cake systems could also target  functionalities such as wave redirection (by metasurfaces) and energy harvesting. To cater for such applications, the factorized Bloch solution could for instance be deployed to expose the power flow through a layer-cake system. Specifically, on substituting~\eqref{factusig}--\eqref{fields2} into~\eqref{PV1} one can readily compute the averaged Poynting vector in each unit cell. An example of such calculation is shown in Fig.~\ref{poynt}, which plots the distribution of averaged Poynting vectors in the rainbow trap shown in Fig.~\ref{RT_dispersion}(a) at four excitation frequencies. The results in particular demonstrate that in the strips featuring significant $\langx\boldsymbol{\mathcal{P}}_{\!\!\text{\tiny T}}\rangx$ magnitudes, the power flow inside the operating range of the rainbow trap ($\omega\shh 2,3,4$) is largely \emph{redirected}, in the form of evanescent waves, either ``up'' or ``down'' the layer-cake system. On recalling the upward direction of the incident plane wave shown in Fig.~\ref{RT_dispersion}(a), the redirection of power flow ``down'' the layer-cake system at $\omega=4$ can be loosely related to the phenomenon of \emph{negative refraction} in that the $\xi_2$-component of the Poynting vector in lamina $j\shh 4$ of the rainbow trap (the principal energy diverter) reverses sign relative to that of the incident wave. More generally, we observe from both  Fig.~\ref{RT_validation} and Fig.~\ref{poynt} that the laminae with large displacement or power flow magnitudes do not necessarily feature band gaps that contain the excitation frequency; a behavior that exposes the limitations of rainbow trap (or energy harvester) design based exclusively of the spectral properties of the mother periodic media.

\begin{figure}
\centering{\includegraphics [width=0.65\textwidth ]{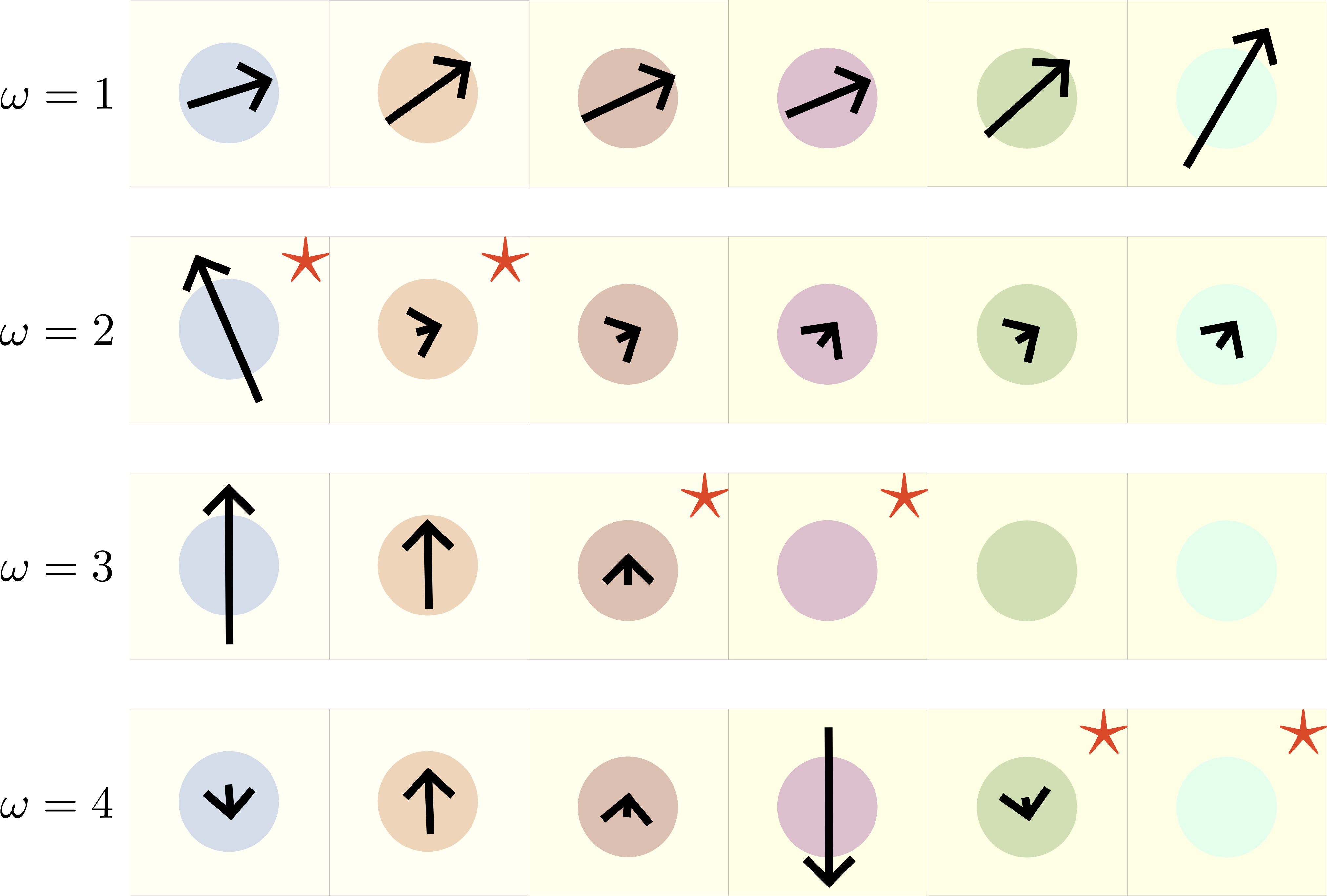}}
\caption{Averaged Poynting vector $\langx\boldsymbol{\mathcal{P}}_{\!\!\text{\tiny T}}\rangx$ (drawn to scale) in each unit cell of the rainbow trap featured in Fig.~\ref{RT_dispersion}. In the diagrams, unit cells whose respective bandgaps contain the excitation frequency are indicated by the ``$\star$'' symbol.} \label{poynt}
\end{figure}

%----------------------------------------------------------------------------------------------------------------------%
\section{Summary and outlook}
%----------------------------------------------------------------------------------------------------------------------%

In this work, we consider scattering of scalar plane waves by a heterogeneous layer that is periodic in the direction parallel to its boundary. On describing the latter as a composite made of periodic strips i.e. laminae, we pursue a solution of the scattering problem by invoking the concept of propagator matrices and that of Bloch eigenstates featured by the unit cell of each lamina. The featured Bloch eigenstates solve the quadratic eigenvalue problem (QEVP) that seeks a complex-valued wavenumber normal to the layer boundary given (i) the excitation frequency, and (ii) real-valued wavenumber parallel to the boundary -- that is conserved throughout the system. Spectral analysis of the QEVP reveals sufficient conditions for discreteness of the eigenspectrum and the fact that all eigenvalues come in complex-conjugate pairs. By way of the factorization brought about by the propagator matrix approach, we demonstrate explicitly that the contribution of individual eigenvalues (and affiliated eigenmodes) to the solution diminishes exponentially with absolute value of their imaginary part, which then forms a systematic platform for truncation of the factorized Bloch solution. The proposed methodology caters for the optimal design of rainbow traps, energy harvesters and metasurfaces, whose ability to manipulate waves is controlled not only by the individual dispersion and impedance characteristics of the component laminae, but also by the ordering and generally fitting of the latter into a composite layer. By providing a natural basis for the expansion of wavefields in layer-cake periodic systems, the factorized Bloch solution can be viewed as a model order reduction that features fast convergence and modular evaluation structure facilitating the optimal design of this class of wave manipulation devices. Additional features of the proposed approach include e.g. natural distinction between the propagating and evanescent waves, and effective tracking of power flow within a layer-cake system. Inherently, this study lends itself to a variety of extensions, including an account for periodic lattices (i.e. periodically ``perforated'' media), elastic waves, and the like treatment of 3D wave scattering by periodic layers comprising one-dimensional sequences of the ``sheets'' of 3D periodic media. 

%----------------------------------------------------------------------------------------------------------------------%
%\ack{The corresponding author kindly acknowledges partial funding provided by the endowed Shimizu Professorship. Thanks are extended to Minnesota Supercomputing Institute for the support during the course of this investigation.}

%----------------------------------------------------------------------------------------------------------------------%

%----------------------------------------------------------------------------------------------------------------------%
\bibliographystyle{plain}\bibliography{refs} 
%----------------------------------------------------------------------------------------------------------------------%

%----------------------------------------------------------------------------------------------------------------------%
\end{document}